\documentclass[sigplan,10pt,nonacm]{acmart}
\renewcommand\footnotetextcopyrightpermission[1]{}
\usepackage{amsmath}
\usepackage{graphicx}
\usepackage{xcolor}
\usepackage{listings}
\usepackage{xspace}
\usepackage{algpseudocode}
\usepackage{enumitem}
\usepackage{multirow}
\usepackage{tabularx, booktabs}
\usepackage{subcaption}
\usepackage{makecell}
\usepackage[scaled=0.85]{beramono}
\usepackage{pgfplots}
\usepackage{tikz}
\usepackage{algorithm}
\usepackage[skins]{tcolorbox}

\usepackage{pifont}
\usepackage{pgfgantt}
\usepackage{bibunits}

\definecolor{csawblue}{RGB}{0,0,139}

\renewcommand{\paragraph}[1]

\newcommand{\stitle}[1]{\noindent{{\bf #1}}}

\definecolor{dkgreen}{rgb}{0,0.6,0}
\definecolor{gray}{rgb}{0.5,0.5,0.5}
\definecolor{mauve}{rgb}{0.58,0,0.82}
\definecolor{awesome}{rgb}{1.0, 0.13, 0.32}
\definecolor{beaver}{rgb}{0.62, 0.51, 0.44}
\definecolor{carrotorange}{rgb}{0.93, 0.57, 0.13}
\definecolor{chocolate}{rgb}{0.82, 0.41, 0.12}
\definecolor{copper}{rgb}{0.72, 0.45, 0.2}
\definecolor{crimsonglory}{rgb}{0.75, 0.0, 0.2}

\theoremstyle{definition}

\newtheorem{theorem}{Theorem}

\newtheorem{definition}{Definition}

\lstdefinestyle{base}{
  language=C,
  emptylines=1,
  breaklines=true,
  basicstyle=\ttfamily\color{blue},
  moredelim=[is][\color{red}]{@}{@},
}

\usepackage[normalem]{ulem}

\begin{document}
\usepgfplotslibrary{groupplots}

\newcommand{\projectname}{{TIDE-MC}\xspace}

\title{\projectname: Two-Sided Interpolative Decomposition for Billion-Scale GPU Matrix Completion}
\author{%
\parbox{0.99\textwidth}{%
\centering
\normalsize
Chengying Huan$^{1}$,
Yubo Wang$^{2}$,
Pinhuan Wang$^{1}$,
Lizheng Chen$^{1}$,
Jie Zhang$^{3}$\\[2pt]
Fangxin Liu$^{4}$,
Qing Wang$^{1}$,
Ruixuan Liu$^{5}$,
Shaonan Ma$^{6}$,
Mingxing Zhang$^{7}$\\[2pt]
Zhibin Wang$^{1}$,
Rong Gu$^{1}$,
Guihai Chen$^{1}$,
Chen Tian$^{1}$\\[5pt]
\normalsize
$^{1}$Nanjing University \quad
$^{2}$Northeastern University (China) \quad
$^{3}$Peking University\\[2pt]
$^{4}$Shanghai Jiao Tong University \quad
$^{5}$Dalian Maritime University \quad
$^{6}$Qiyuan Lab \quad
$^{7}$Tsinghua University
}}

\renewcommand{\shortauthors}{Huan et al.}

\begin{abstract}

Matrix completion supports large-scale recommendation and scientific computing, yet existing GPU solvers commonly assume that the observed matrix or its dense factors fit in device memory.
On real workloads, this assumption leads to out-of-memory failures or severe PCIe overhead under naive paging.

We present \projectname, a bounded-memory GPU framework built on Two-Sided Interpolative Decomposition (TSID).
TSID uses a sampled template submatrix as an anchor for reconstructing the full low-rank matrix, allowing computation and storage to scale with the template and active data chunks rather than the complete matrix.
\projectname realizes this formulation through two execution stages.
First, a conflict-free synchronization engine recovers the template using parallel factorization and hierarchical gradient aggregation.
Second, a chunked reconstruction pipeline extends the recovered template to the remaining matrix while overlapping PCIe transfers with GPU computation.
An asymmetric gradient-clipping scheme stabilizes mixed-precision Tensor Core execution.

Across 15 benchmarks, \projectname completes workloads that cause existing GPU solvers to run out of memory.
Compared with the evaluated state-of-the-art baselines, it achieves up to $11{,}647\times$ speedup, reduces peak memory usage by up to $8.5\times$, and lowers reconstruction error by up to $99.7\%$.
These results show that template-anchored decomposition and stage-specific GPU execution can scale matrix completion beyond device-memory capacity.

\end{abstract}

\pagestyle{plain}

\maketitle

\section{Introduction}

Matrix completion has become a cornerstone of modern data-driven applications, such as collaborative filtering~\cite{baby2024online,feng2020fusion,borgs2017thy}, genomic data analysis~\cite{cai2016structured,natarajan2014inductive,kapur2016gene,chi2013genotype,mongia2019mcimpute}, network reconstruction~\cite{sadeghi2019signal}, graph neural networks~\cite{10.1145/3604915.3610654}, information processing systems~\cite{NEURIPS2024_626ab938}, image reconstruction~\cite{Wang_2016_CVPR,hu2012fast,cao2014image,krajewska2024matrix,krajewska2026randomized,radhakrishnan2022simple,shin2014calibrationless}, industrial recommendation systems~\cite{bennett2007netflix,koren2009matrix}, and multi-label learning~\cite{Li_Zhu_Wang_Zhang_Lai_Wang_2025}. Formally, matrix completion involves recovering a complete matrix \(\widehat{R}\) from a sparse set of observations \(\Omega\) of an underlying matrix \(R \in \mathbb{R}^{m \times n}\). In practice, this is typically achieved by exploiting low-rank structure,e.g., 
approximating \(\widehat{R}\) or a candidate completion with low-dimensional
matrices \(P\) and \(Q\)
while optimizing only over the observed entries.

In real-world scenarios, matrix completion often operates on enormous user-item matrices or high-dimensional data tables, of which only a tiny fraction of entries are observed. Such matrices are typically \emph{sparse} and approximately \emph{low-rank}~\cite{babacan2012sparse,tao2011recovering}. For example, users in recommendation systems rate only a small subset of items~\cite{NetflixPrize,yahoo2023dataset}; in computer vision, sensor measurements may miss pixels or views; and in genomic or network analysis, only a subset of possible interactions is measured~\cite{serra2018robust}. These shared characteristics make matrix completion both practically important and increasingly difficult to scale as data dimensions continue to grow.

As these applications grow to industrial and scientific scales, matrix completion becomes a large-scale systems and HPC challenge in addition to an optimization problem. Modern deployments must handle matrices with billions of potential entries under strict resource constraints, which further amplifies the need for scalable techniques in large-scale data processing systems~\cite{narayanan2018gandiva,jeon2021zico,rajbhandari2020zero,shoeybi2021megatron,zheng2022alpa}. For example, Netflix manages user--item rating matrices with over 100 million observed entries~\cite{bennett2007netflix}, while Yahoo! Music processes ratings from 200{,}000 users across 136{,}000 songs, amounting to over 27 billion potential entries~\cite{yahoo2023dataset}. Under such scale, practical matrix completion must operate within limited memory while sustaining high throughput and low latency for near real-time imputation.

Traditional matrix completion methods largely fall into two families: nuclear norm minimization (NN)-centric approaches~\cite{krajewska2024matrix,hu2012fast} and singular value decomposition (SVD)-centric approaches~\cite{abdi2007singular,2014Matrix,2017Matrix,2014Performance}. NN-centric methods recover missing entries by minimizing the nuclear norm of the matrix, either directly~\cite{shin2014calibrationless} or after reducing the matrix via column selection, as in CSNN~\cite{krajewska2024matrix}. SVD-centric methods operate on low-rank factorizations (e.g., \(A = U \Sigma V^T\)) and use iterative solvers such as CSPGD and CGM~\cite{zier2024adapting}, kernel-based approximations like NTK~\cite{radhakrishnan2022simple}, or block-wise partitioning as in MVGMC~\cite{koohi2019parallel} to improve scalability. Table~\ref{tab:all} summarizes the time and space complexity of the above methods and shows that, despite these optimizations, their costs still scale unfavorably with the matrix size, limiting their applicability to large-scale matrix completion.

\begin{table}[t]
\centering
\caption{Comparison of time and resident space complexity for prior systems and \projectname. The matrix is in \(\mathbb{R}^{n \times n}\); \(\alpha\) is the template sampling rate (\(0.1\)--\(0.3\)); \(\delta\) is the observed-entry fraction(\(8\%\)--\(20\%\)); \(k\) is the latent rank, approximately \(\alpha n\); \(t_1,t_2\) are the iteration counts of Stage~I and Stage~II(ranging from 100 to 10{,}000); and \(\Delta c\) is the chunk width. For \projectname, space denotes the peak resident GPU working set rather than the total streamed input size.}
\label{tab:all}
\small
\resizebox{\linewidth}{!}{
\begin{tabular}{|c|c|c|c|c}
\hline
\textbf{Type} & \textbf{Name} & \textbf{Time} & \textbf{Space} \\
\hline
\multirow{2}{*}{\begin{tabular}[c]{@{}c@{}}NN-\\ Centric\end{tabular}} 
& NN~\cite{shin2014calibrationless} & $\mathcal{O}(k n^3)$ & $\mathcal{O}(n^2)$ \\
& CSNN~\cite{krajewska2024matrix} & $\mathcal{O}(\alpha^2n^3)$ & $\mathcal{O}(\alpha n^2)$  \\
\hline
\hline
\multirow{4}{*}{\shortstack{SVD-\\ Centric}} 
& CSPGD~\cite{krajewska2024matrix} & $\mathcal{O}(\alpha n^3)$ & $\mathcal{O}(\alpha^2 n^3)$ \\
& CGM~\cite{krajewska2024matrix} & $\mathcal{O}(kn^2)$ & $\mathcal{O}(n^2)$ \\
& NTK~\cite{radhakrishnan2022simple} & $\mathcal{O}(k n^2)$ & $\mathcal{O}(\delta n^2)$ \\
& MVGMC~\cite{koohi2019parallel} & $\mathcal{O}(k n^3(1 + \delta))$ & $\mathcal{O}(n^2)$   \\
\hline
\hline
\multicolumn{2}{|c|}{\textbf{Our System} \textbf{\projectname}}
& $\mathcal{O}(k\delta n^2(t_2+\alpha t_1))$
& $\mathcal{O}(k(1+\alpha)n + |\Omega_{\Delta c}| + k\Delta c)$ \\
\hline
\end{tabular}
}

\end{table}

\noindent\textbf{Limitations as Computing Systems.}
Although NN-centric and SVD-centric methods differ mathematically, modern implementations of both remain monolithic whole-matrix engines: they assume that the full sparse matrix, or large dense factors or Gram matrices, reside entirely in GPU high-bandwidth memory (HBM) and organize computation around whole-matrix operators. This design is the main source of scalability problems, because it tightly couples the \emph{logical} problem size to the \emph{physical} memory capacity of the hardware. Once the matrix exceeds HBM capacity, such systems either hit an out-of-memory error or become stall-bound on PCIe under naive paging or offloading.

\noindent\textbf{Key Insight.}
Rather than proposing another optimizer, we seek a resource-reduction primitive for matrix completion. We introduce Two-Sided Interpolative Decomposition (TSID), which provides an exact-rank skeleton identity showing that a well-conditioned sampled template can serve as an algebraic anchor for reconstructing the ideal rank-\(k\) component, while practical error is governed by the rank-\(k\) residual, template recovery error, and extension error. TSID therefore reduces a large completion problem to a smaller resident subproblem and yields a template--extension structure.

\noindent\textbf{Our Design.} Building on this abstraction, we use TSID as the foundation of a submatrix-centric scheme for matrix completion: instead of requiring the entire matrix to reside in HBM, we keep only a template submatrix resident on the GPU and stream the remaining columns (or rows) as chunks.
This decouples the logical matrix size from the resident working set: the working set scales linearly with the matrix dimension, while the logical problem remains quadratic. On top of this scheme, we present \projectname, a framework for matrix completion that exposes bounded-memory completion as a reusable service. \projectname replaces the monolithic execution model of prior solvers with a two-stage bounded-memory pipeline~\cite{li2014ps,narayanan2018gandiva,zheng2022alpa,narayanan2019pipedream,rajbhandari2020zero,shoeybi2021megatron,jeon2021zico} that balances compute density and memory bandwidth.
Our specific contributions are:
\begin{itemize}[leftmargin=*,topsep=0pt,itemsep=0pt,parsep=0pt]
  \item \textbf{TSID-Guided Abstraction.}
  We elevate TSID from a mathematical result to a resource-reduction primitive for matrix completion. Guided by TSID, we introduce a two-stage design: Stage~I recovers a template submatrix, and Stage~II extends it to the full matrix under a fixed memory budget.

  \item \textbf{Bounded-Memory Two-Stage Execution Framework.}
We design a TSID-guided execution framework. Stage~I recovers the template submatrix using a conflict-free parallel SGD-based RMF kernel with gradient aggregation, reducing atomic contention and mapping irregular updates to dense Tensor Core tiles. Stage~II reconstructs the full matrix via roofline-guided chunked alternating least squares (ALS), using double buffering and autotuned chunk sizes to overlap PCIe transfers with computation and keep GPUs near the roofline (Section~\ref{sec:opt1} \& Section~\ref{sec:opt2}).

  \item \textbf{Numerical Guardrail for Mixed-Precision Execution.}
  To accelerate calculations, we use FP16/FP32 mixed precision in both stages, and for safe exploitation, we utilize asymmetric gradient clipping as a numerical guardrail. It caps gradient explosion while allowing shrinkage, preventing FP16 overflow and NaNs on power-law workloads and enabling stable Tensor Core execution(Section~\ref{sec:opt3}).

  \item \textbf{Comprehensive Evaluation at Billion Scale.}
We benchmark \projectname\ against state-of-the-art matrix completion methods on 8 synthetic datasets and 7 real datasets. \projectname delivers up to $11{,}647\times$ speedup, $8.5\times$ lower peak memory usage, and 99.7\% lower reconstruction error on billion-scale matrices, whereas competing frameworks either time out or run out of memory (Section~\ref{sec:eva}).
\end{itemize}

\section{Background}
\label{sec:background}


Matrix completion aims to recover the missing entries of a partially observed matrix. Formally, we have the definition:

\begin{definition}[Matrix Completion]
Given an incomplete matrix \( R \in \mathbb{R}^{m \times n} \) with observed
entries indexed by \(\Omega\), the matrix completion problem seeks to recover
a full matrix \(\widehat{R}\in\mathbb{R}^{m\times n}\) such that \(\widehat{R}_{ij}=R_{ij}\) for all \((i,j)\in\Omega\), and \(\widehat{R}\) accurately estimates the unobserved entries.
\end{definition}


Traditional matrix completion methods largely fall into two categories: nuclear norm minimization (NN)-centric methods and singular value decomposition (SVD)-centric methods in practice.

\subsection{NN-Centric Methods}

NN-centric approaches, such as nuclear norm minimization (NN) and
Column-Selected Nuclear Norm minimization (CSNN), recover low-rank matrices by convex rank surrogates. 
Figure~\ref{current}(a) illustrates a NN-centric pipeline: given
\(R\), the solver iteratively computes gradients over the observed entries and
applies SVT-based shrinkage, while assuming that the full matrix resides in device memory.

\noindent\textbf{NN}~\cite{shin2014calibrationless} relaxes rank
minimization by replacing the rank with the sum of singular values. Given
\(\Omega \subseteq [n_1]\times[n_2]\), NN solves
\begin{equation*}
    \widehat{R}
= \arg\min_{X \in \mathbb{R}^{n_1 \times n_2}} \|X\|_{*}
\quad \text{s.t.} \quad
\mathcal{R}_{\Omega}(X - R) = 0,
\end{equation*}
where \(\|X\|_{*}\) denotes the nuclear norm and \(\mathcal{R}_{\Omega}\)
retains observed entries. This problem can be formulated as a semidefinite
program~\cite{parrilo2003semidefinite,nesterov2000semidefinite} and is often
solved by singular value thresholding (SVT)~\cite{cai2010singular}, which
iteratively applies gradients and singular-value shrinkage. While NN offers
strong recovery guarantees, repeated full-matrix SVDs
\cite{lu2015generalized,cai2010singular} lead to high memory usage, limited
GPU utilization, and poor scalability.

\noindent\textbf{CSNN}~\cite{krajewska2024matrix,krajewska2026randomized} reduces this cost by applying NN to a uniformly sampled subset of \(d\) columns. Let \(I \subseteq [n_2]\) and \(C = R_{:,I} \in \mathbb{R}^{n_1 \times d}\). CSNN solves
\begin{equation*}
\widehat{C} = \arg\min_{X \in \mathbb{R}^{n_1 \times d}} \|X\|_{*}
\quad \text{s.t.} \quad \mathcal{R}_{\Omega_I}(X - C) = 0,
\end{equation*}
where \(\Omega_I = \{ (i,j) \in \Omega \mid j \in I \}\), and then reconstructs the full matrix via least squares~\cite{golub1980analysis}:
\begin{equation*}
    \widehat{Z} = \arg\min_{Z \in \mathbb{R}^{d \times n_2}}
  \frac{1}{2}\Bigl\|\mathcal{R}_{\Omega}(R)
  - \mathcal{R}_{\Omega}(\widehat{C}\,Z)\Bigr\|_F^2,
  \widehat{R} = \widehat{C}\ \widehat{Z}.
\end{equation*}
Restricting SDP computation~\cite{parrilo2003semidefinite,nesterov2000semidefinite}
to a tall-and-skinny block (\(d \ll n_2\)) lowers runtime and memory, and
enables GPU-accelerated conic solvers~\cite{zier2024adapting}.

\noindent\textbf{Limitations.}
From a systems viewpoint, NN and CSNN remain expensive and difficult to scale. NN incurs quadratic space and cubic time complexity, while CSNN still depends on dense subproblems with substantial memory overhead. Moreover, neither method provides a streamable or chunked execution model. Consequently, on large datasets such as Netflix and Yahoo! Music (Table~\ref{tab:all}), state-of-the-art NN-centric implementations either time out or exhaust GPU memory.

\begin{figure}[t]
\centering
\subfloat[NN-centric.]
{
    \begin{minipage}[b]{0.385\linewidth}
        \centering
        \includegraphics[width=\linewidth]{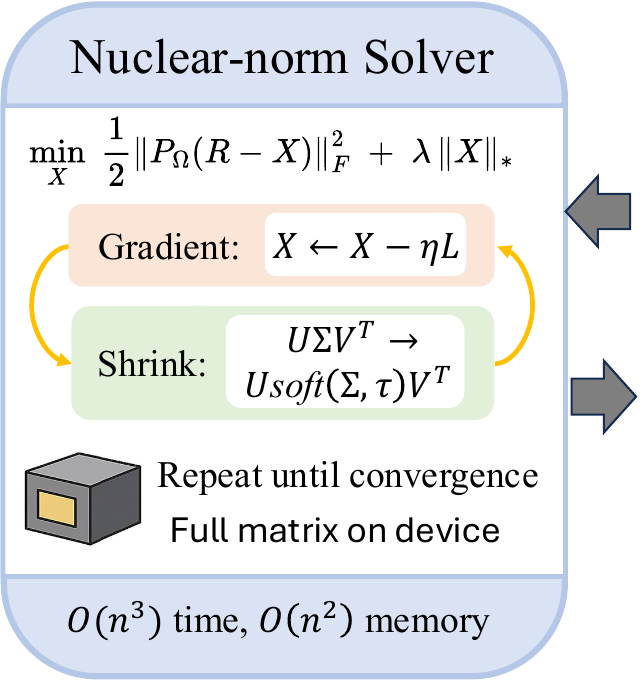}
    \end{minipage}
}
\hspace{-0.04\linewidth}
\subfloat[SVD-centric.]{
    \begin{minipage}[b]{0.579\linewidth}
        \centering
        \includegraphics[width=\linewidth]{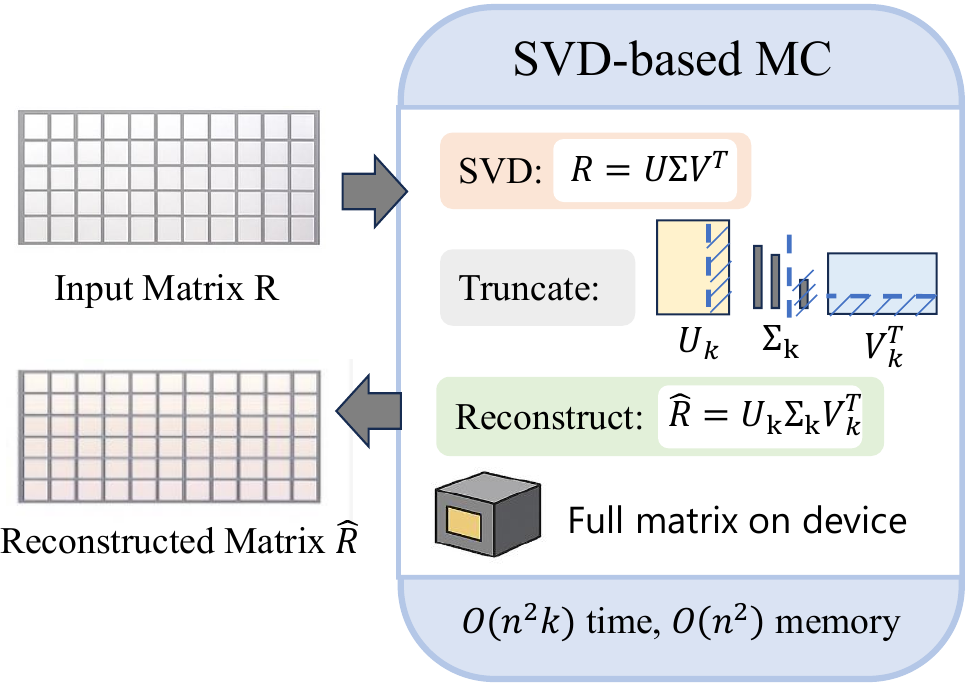}
    \end{minipage}
}
\caption{Workflow of NN-centric and SVD-centric matrix completion methods, both requiring the full matrix to reside in device memory during execution.}
\label{current}
\end{figure}

\subsection{SVD-Centric Methods}

SVD-centric methods trade some accuracy for higher computational efficiency.
Representative examples include column-selected proximal gradient descent (CSPGD), conditional gradient methods (CGM), NTK-based matrix completion (NTKMC), and multi-view graph-based matrix completion (MVGMC). Figure~\ref{current}(b) illustrates a typical SVD-centric workflow, in which a low-rank decomposition \(R \approx U \Sigma V^T\) of the entire matrix is used to reconstruct the completed matrix or enable further refinement, again assuming full-matrix residency in device memory.

\noindent\textbf{CSPGD}~\cite{krajewska2024matrix} accelerates CSNN by
solving a proximal gradient formulation~\cite{ye2020decentralized} on the
sampled submatrix:
\begin{equation*}
\min_{X\in\mathbb{R}^{n_1\times d}} \; \frac{1}{2}\left\|\mathcal{R}_{\Omega_I}(X - C)\right\|_F^2 + \lambda\,\|X\|_{*},
\end{equation*}
with iterative updates
\begin{equation*}
X^{(t+1)} = \mathrm{SVT}_{\eta \lambda}\!\ \Bigl(
X^{(t)} - \eta\,\nabla\!\Bigl[\frac{1}{2}\left\|\mathcal{R}_{\Omega_I}(X^{(t)}-C)\right\|_F^2\Bigr]
\Bigr),
\end{equation*}
where \(\mathrm{SVT}_{\tau}\) is singular value thresholding~\cite{cai2010singular}.
CSPGD avoids a full SDP~\cite{pataki2000geometry} and benefits from
GPU-accelerated tensor operations, but still requires repeated SVDs of
\(n_1\times d\) and sublinear convergence.

\noindent\textbf{CGM.}
The conditional gradient method~\cite{krajewska2024matrix}, a Frank–Wolfe
variant~\cite{wai2017decentralized,clarkson2010coresets}, optimizes
\(
f(X) = \frac{1}{2}\|\mathcal{R}_{\Omega_I}(X - C)\|_F^2
\)
over a nuclear-norm ball. At each iteration, it computes the top singular
pair \((u^{(t)}, v^{(t)})\) of \(\nabla f(X^{(t)})\) and updates
\(
X^{(t+1)} = \bigl(1 - \gamma_t\bigr) X^{(t)} + \gamma_t\,\tau\, u^{(t)}{v^{(t)}}^T.
\)
GPU-accelerated SVD makes each step efficient, but the
\(\mathcal{O}(1/t)\) rate and growing intermediate ranks can still be costly
for high-precision recovery.

\noindent\textbf{NTKMC}~\cite{radhakrishnan2022simple} derives a
problem-specific Neural Tangent Kernel and casts completion as kernel ridge
regression with closed-form predictor
\(
f^*(x) = K(x, X)\,\bigl(K(X,X) + \lambda I\bigr)^{-1} Y,
\)
where \(X\) and \(Y\) are coordinates and observed values in \(\Omega\).
Training inverts a \(|\Omega|\times|\Omega|\) Gram matrix
(\(\mathcal{O}(|\Omega|^3)\) time, \(\mathcal{O}(|\Omega|^2)\) memory), and
each prediction costs \(\mathcal{O}(|\Omega|)\). Dense kernel operations limit
scalability and GPU efficiency.

\noindent\textbf{MVGMC}~\cite{koohi2019parallel} augments
matrix completion with multi-view information, e.g., rating and item-similarity matrices, to handle sparse data and cold-start scenarios
\cite{narayanan2008robust,ocepek2015improving}. The optimization couples
rating reconstruction and similarity consistency:
\begin{equation}
    \min_{\Omega} \; \| R - P Q^T \|_F + \lambda \| P \|_F + \beta \| S - Q Q^T \|_F,
\label{eqn-multi-view}
\end{equation}
where \(\| S - Q Q^T \|_F\) encourages item factors in \(Q\) to match the
similarity matrix \(S\), and \(\lambda,\beta\) are the regularization parameters.
The similarity matrix is built from item attributes:
\(
S_{ii'} = \frac{1}{J} \sum_{j=1}^J \mathbb{I}(a_{ij} = a_{i'j}),
\)
where \(J\) is the number of attributes, and \(\mathbb{I}\) is the indicator
function.

\noindent\textbf{Limitations.}
Despite differing mathematically, these systems behave as whole-matrix
engines: they require the full matrix (or dense factors) to fit in device
memory or rely on paging that is oblivious to low-rank structure. They lack a task-aware mechanism for memory-efficient execution, motivating our search for a resource-reduction primitive that enables a submatrix-centric design with low space cost.

\section{Design Principle}
\label{sec:opportunity}

As discussed in Section~\ref{sec:background}, traditional matrix completion systems face a fundamental scalability bottleneck. In real-world scenarios, the target matrix is typically very sparse and approximately low-rank~\cite{babacan2012sparse,tao2011recovering}. As matrix dimensions grow, storage and intermediate structures often scale as \(\mathcal{O}(n^2)\), while the associated computation ranges from \(\mathcal{O}(n^2)\) to \(\mathcal{O}(n^3)\). At scales of hundreds of thousands of rows or columns, even storing dense factors or Gram matrices can exceed GPU memory. The core systems challenge is thus to carry out large scale matrix completion without assuming that the full matrix fit in accelerator memory.

A tempting approach is to treat matrix completion as a generic sparse computation and rely on paging or swapping: keep most data in CPU memory or on disk and fault in segments on demand. However, without semantic knowledge of the completion task, such mechanisms move arbitrary subsets of the matrix at the page granularity. They neither exploit low-rank structure nor respect the algebraic dependencies between columns; thus, they spend most of their time shuffling data over PCIe rather than computing. Therefore, we algebraically \emph{reduce} the completion problem such that only a small, well-chosen part of the matrix needs to reside on the GPU.

\begin{table}[t]
\centering
\caption{Illustrative comparison between the full matrix and an arbitrary subset on a synthetic low-rank matrix completion workload. (\(\alpha = 0.01\))}
\label{tab:tsid_observation}
\footnotesize
\setlength{\tabcolsep}{5pt}
\begin{tabular}{lcc}
\toprule
\textbf{Metric} & \textbf{Full matrix $R$} & \textbf{Arbitrary subset $R_s$} \\
\midrule
Shape & $m \times n$ & $\alpha( m \times n)$ \\

Top-$k$ energy retained 
& $1.00$ 
& $0.90$ \\

Subspace similarity to $R$ 
& $1.00$ 
& $0.80$ \\

Resident working set 
& high 
& low \\
\bottomrule
\end{tabular}
\end{table}
\begin{figure*}[t]
  \centering
  \includegraphics[width=0.98\linewidth]{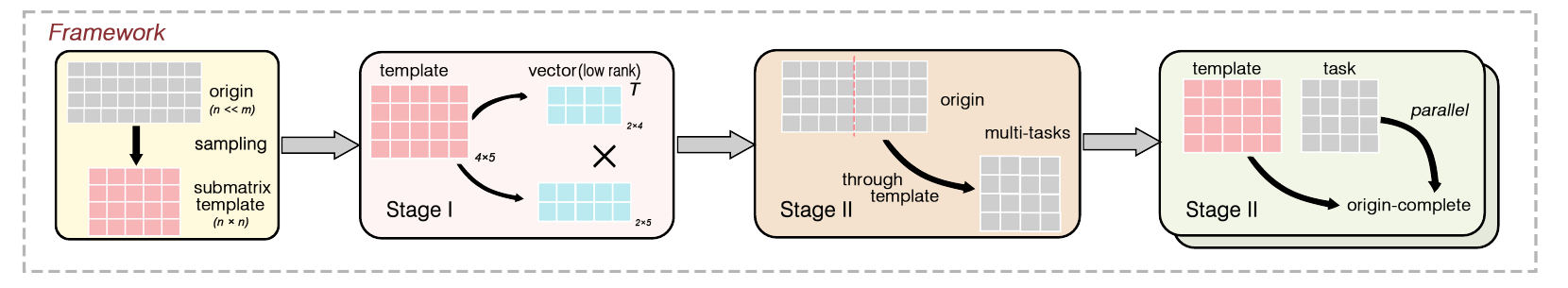}
  \caption{Overall workflow of \projectname: submatrix sampling and recovery (Stage~I) and chunk-based ALS reconstruction (Stage~II).}
  \label{algorithm}
\end{figure*}

\noindent\textbf{Observation.} 
Classical interpolative decomposition (ID) results show that \textit{low-rank matrices can be approximated by subsets of their rows or columns}. Table~\ref{tab:tsid_observation} provides empirical evidence that a suitably chosen submatrix can preserve the dominant low-rank structure of \(R\), suggesting that recovery of the full matrix may be approximated by first recovering a smaller representative submatrix and then extending it. This motivates us to develop a \textbf{Two-Sided Interpolative Decomposition (TSID)} for matrix completion. Under mild assumptions, TSID states that if sufficiently informative subsets of rows and columns are selected from \(R\), then the resulting submatrix \(R^*\) both preserves the low-rank structure of \(R\) and serves as an effective \emph{template}; furthermore, recovering \(R^*\) from the observed entries is, up to standard low-rank approximation error, equivalent to recovering \(R\).

\subsection{TSID Theoretical Basis}

\textbf{Notation and Setup.}
Let \(R\in\mathbb{R}^{m\times n}\) denote the full target matrix, and let
\[
R = U\,\Sigma\,V^T
  = \begin{bmatrix}U_1 & U_2\end{bmatrix}
    \begin{bmatrix}\Sigma_1 & 0\\[4pt]0 & \Sigma_2\end{bmatrix}
    \begin{bmatrix}V_1^T\\[4pt]V_2^T\end{bmatrix}
\]
be its singular value decomposition, where \(\Sigma_1\in\mathbb{R}^{k\times k}\) contains the top \(k\) singular values, \(U_1\in\mathbb{R}^{m\times k}\), and \(V_1\in\mathbb{R}^{n\times k}\). We write the ideal rank-\(k\) component as
\[
{
R_k = U_1\Sigma_1V_1^T .
}
\]
Select index sets \(I\subseteq\{1,\dots,m\}\) and \(J\subseteq\{1,\dots,n\}\), each of size \(k\), such that the ideal core
\[
{
R_k^* = (R_k)_{I,J}
       = U_1(I,:)\,\Sigma_1\,V_1(J,:)^T
}
\]
is nonsingular. Define the left and right interpolation matrices
\[
{
W_L = U_1\bigl(U_1(I,:)\bigr)^{-1}\in\mathbb{R}^{m\times k},
\qquad
W_R = V_1\bigl(V_1(J,:)\bigr)^{-1}\in\mathbb{R}^{n\times k}.
}
\]
The superscript \(*\) denotes the sampled template. In the ideal theorem below it refers to the fully specified rank-\(k\) core \(R_k^*\); in the practical matrix-completion algorithm, the corresponding observed template is only partially observed and must be estimated from data.

\noindent\textbf{Theorem and Proof of TSID.}
The theorem below is an exact-rank skeleton/CUR identity.  Its role is to justify why a well-conditioned template can serve as an algebraic anchor for bounded-memory extension.
\begin{theorem}[Ideal TSID skeleton identity]
\label{thm:tsid-skeleton}
If \(R_k^*=(R_k)_{I,J}\) is nonsingular, then the interpolation induced by the selected rows and columns exactly reconstructs the ideal rank-\(k\) component:
\[
{
\widetilde{R}_k
 = W_L\,R_k^*\,W_R^T
 = R_k .
}
\]
For a matrix \(R\) with nonzero residual \(R-R_k\), this identity applies only to the ideal component \(R_k\). The resulting error depends on row/column sampling, core conditioning, template recovery, and extension error.

\end{theorem}

\begin{proof}
By the definitions above,
\[
{
\begin{aligned}
\widetilde{R}_k
&= W_L\,R_k^*\,W_R^T \\
&= U_1\bigl(U_1(I,:)\bigr)^{-1}
   \bigl(U_1(I,:)\Sigma_1V_1(J,:)^T\bigr)
   \bigl(V_1(J,:)\bigr)^{-T}V_1^T \\
&= U_1\Sigma_1V_1^T
 = R_k .
\end{aligned}
}
\]
This proves the exact identity for the ideal rank-\(k\) component. If \(R\) is not exactly rank \(k\), then \(R=R_k+(R-R_k)\), and the identity above does not reconstruct the residual \(R-R_k\).
Therefore, TSID should be interpreted as a skeleton-based abstraction for reducing the resident problem size, not as an end-to-end proof of optimal sparse matrix completion.
\end{proof}

\subsection{From TSID to a Two-Stage Framework}~\label{sec:TSiD-stage}
\vspace{-2mm}

Given the TSID reduction in Theorem~\ref{thm:tsid-skeleton}, matrix completion is reorganized into two bounded-memory
stages: (i) recovering a compact \emph{template} submatrix and
(ii) extending this template to reconstruct the full target matrix.

In \textbf{Stage I (template recovery)}, a sampled template submatrix \(R^*\) is
factorized using regularized matrix factorization (RMF). Let \(\Omega^*\subseteq\Omega\) denote the observed entries that fall inside the sampled template. The submatrix is
parameterized as \(\widehat{R^*} \approx P Q^{\top}\), with
\(P \in \mathbb{R}^{m\times k}\) and \(Q \in \mathbb{R}^{\alpha n\times k}\), by
solving
\begin{equation}
\label{eq:rmf-obj}
\min_{P,Q}\;
  \sum_{(j,i)\in\Omega^*}
    \bigl(R^*_{ji} - p_j^\top q_i\bigr)^2
  + \lambda \bigl(\|P\|_F^2 + \|Q\|_F^2\bigr),
\end{equation}
where \(p_j\),\(q_i\) are rows of \(P\),\(Q\), and \(\lambda\)
controls \(\ell_2\) regularization. This objective is optimized by stochastic
gradient descent.

In \textbf{Stage~II (extension)}, the recovered template \(\widehat{R^*} \approx P
Q^{\top}\) is kept fixed, and the remaining columns of the original matrix are
reconstructed by solving regularized least-squares problems conditioned on
this template. The objective decomposes over columns, which makes it amenable
to chunked, pipelined execution under a fixed memory budget; the concrete
execution engine is described in Sections~\ref{sec:design} and~\ref{sec:opt}.

\noindent\textbf{Practical Approximation Error.}
The identity above is lossless only for the ideal component \(R_k\). To see what changes in the practical setting, write \(R=R_k+E\), where \(E\) is the residual outside the rank-\(k\) component. The sampled template of the full matrix is then \(R^*=R_{I,J}=R_k^*+E_{I,J}\). If Stage~I recovers an estimated template \(\widehat{R}^*\) and Stage~II forms \(\widehat{R}=W_L\widehat{R}^*W_R^T\), then
\[
{
\begin{aligned}
\widehat{R}-R
&= W_L(\widehat{R}^*-R_k^*)W_R^T - E,\\
\|\widehat{R}-R\|_F
&\le \|E\|_F
+ \|W_L\|_2\|W_R\|_2
\left(\|\widehat{R}^*-R^*\|_F+\|E_{I,J}\|_F\right).
\end{aligned}
}
\]
Thus TSID does not imply error-free submatrix reconstruction in general. It shows lossless recovery only under the exact-rank, fully observed, well-conditioned-template assumption. In practical sparse matrix completion, the reconstruction quality depends on the rank-\(k\) residual, the conditioning of the selected rows and columns, the sampling pattern, and the template recovery error.

This formulation also explains why the implementation does not simply set \(P=R^*\) and \(Q=I\). That choice is only a degenerate factorization when the \(k\times k\) template is fully observed, noise-free, and already known. In sparse matrix completion, the sampled template \(R^*\) is only partially observed and may be noisy. Stage~I therefore uses regularized matrix factorization to estimate a low-rank template representation \(\widehat{R^*}\approx P Q^T\) from the observed entries. Stage~II then uses the recovered template representation in a template-conditioned least-squares extension. Thus RMF/ALS-style updates are optimization mechanisms for the practical sparse setting, not mathematical requirements of the exact skeleton identity.

In practical sparse workloads, TSID should be viewed as a systems abstraction for reducing resident memory, not as a standalone guarantee of optimal matrix completion.

\subsection{Complexity}
\label{sec:tsid-complexity}

Let \(R\) be the partially observed input matrix and let \(\Omega\) be its observed-entry set. For an \(n\times n\) matrix with observed-entry density \(\delta\) and template sampling rate \(\alpha\), we have \(|\Omega|=\delta n^2\) and \(|\Omega^*|\approx \alpha\delta n^2\). Let \(k\) be the latent rank, \(t_1,t_2\) the iteration counts of Stage~I and Stage~II, and \(\Delta c\) the chunk width. The two stages require
\[
{
T_{\mathrm{I}} = O(t_1 k|\Omega^*|)
= O(t_1 k\alpha\delta n^2),
\ \
T_{\mathrm{II}} = O(t_2 k|\Omega|)
= O(t_2 k\delta n^2).
}
\]
Thus the total arithmetic work is \(O(k\delta n^2(t_2+\alpha t_1))\). This is still linear in the observed entries; \projectname does not claim sublinear input processing. Its reduction is in peak resident GPU memory: at any time, the system keeps only the template factors and one active chunk,
\[
{
S_{\mathrm{resident}}
=
O\!\left(k(1+\alpha)n + |\Omega_{\Delta c}| + k\Delta c\right),
}
\]
where \(\Omega_{\Delta c}\) is the observed-entry set in the active chunk. Therefore, the resident memory is controlled by the template size and chunk width rather than by the full matrix.

\section{\projectname Framework Overview}

\subsection{Overview}
\label{sec:design}
Section~\ref{sec:opportunity} introduced a TSID-guided Template-Extension decomposition of matrix completion at the algorithmic level. Figure~\ref{algorithm} shows how \projectname
realizes this decomposition as a bounded-memory execution workflow. The key idea is to break the dependence of the resident working set on the logical matrix size: at any time, the GPU holds only a small submatrix or a streaming chunk, rather than the full matrix.

The workflow proceeds in two phases. \projectname first samples a
subset of columns from the input sparse matrix \(R\) to form a compact
submatrix \(R^*\), and recovers \(R^*\) as a low-rank template.
This template captures the dominant structure of \(R\) while remaining
small enough to stay resident in GPU memory. \projectname then treats
the recovered template as a fixed anchor and reconstructs the remaining
columns chunk by chunk, streaming them through the GPU under a fixed
memory budget. In this way, TSID provides the algebraic basis for reducing full-matrix completion to template recovery followed by template-guided extension, while the system realizes this reduction as a submatrix-centric dataflow.

This execution model has three structural properties.
First, it bounds the resident working set to the template and one chunk at a time, enabling completion beyond GPU memory capacity.
Second, it exposes both phases through a common template-guided workflow that can be mapped to GPU linear-algebra kernels.
Third, it separates the Template-Extension abstraction from the stage-specific execution mechanisms, allowing each phase to be optimized independently.

\subsection{Main Challenges}
Although the workflow in Section~\ref{sec:design} establishes a bounded-memory execution path, a straightforward implementation is still insufficient for performance. A direct TSID-guided realization would sample a submatrix \(R^*\), apply an off-the-shelf SGD-based RMF
solver to recover a template, and then run a standard ALS solver to
reconstruct the full matrix conditioned on that template. In practice,
this naive realization still encounters three major bottlenecks.

\noindent\textbf{Challenge I: Atomic Synchronization Overhead (Stage~I).}
During template recovery, RMF must frequently and efficiently update shared latent factors. A naive GPU implementation assigns threads to individual ratings \((j,i)\) and updates user and item factors with atomic operations. On power-law datasets, popular items become hotspots, and thousands of updates contend for the same memory locations. In our profiling, this naive RMF kernel spends over 40\% of its cycles stalled on memory locks, making Stage~I synchronization-bound and unable to fully utilize the GPU.

\noindent\textbf{Challenge II: Stall-Bound Streaming Pipeline (Stage~II).}
Global reconstruction requires streaming large portions of the matrix that do not fit in GPU memory. A naive implementation performs a sequence of blocking steps: copy a chunk over PCIe, compute ALS updates on it, copy it back, and then proceed to the next chunk. On our Stage~II matrix completion microbenchmark, This exposes PCIe latency and limited bandwidth directly on the critical path: GPU SMs spend up to 65\% of runtime idle, waiting for data. Without a dedicated latency-hiding dataflow, Stage~II quickly becomes strictly bandwidth-bound.

\noindent\textbf{Challenge III: Numerical Instability Under Mixed Precision.}
To fully exploit Tensor Cores, both RMF and ALS should run in FP16 or mixed precision. However, naive SGD-based RMF is numerically brittle on sparse, skewed data. The limited dynamic range of FP16, combined with large gradients in early epochs, leads to overflow and divergence. Falling back to FP32 erodes the benefits of specialized hardware, leaving a gap between efficiency and stability.

These challenges motivate the system mechanisms of \projectname. The next section presents how we address them with a conflict-aware synchronization engine, a pipelined chunked reconstruction dataflow,
and a numerical guardrail for stable mixed-precision execution.

\section{\projectname System}
\label{sec:opt}

The Template–Extension execution model in Section~\ref{sec:design} is instantiated
in \projectname as a two-module execution engine.
The \textbf{Template module} is a conflict-aware, hierarchical update engine that
recovers the TSID template under massive GPU parallelism, and the
\textbf{Extension module} is a roofline-driven, pipelined dataflow that
reconstructs the full matrix under a bounded GPU memory budget.
Both modules share a common precision-management layer, i.e., an asymmetric gradient
guardrail that makes aggressive mixed-precision (FP16/FP32) execution numerically
stable.

\subsection{Conflict-Free Synchronization Engine}

\label{sec:opt1}
\begin{table}[t]
\centering
\caption{Submatrix sampling overhead.}
\label{tab:time}
\scriptsize
\setlength{\tabcolsep}{2pt}
\renewcommand{\arraystretch}{0.9}

\begin{subtable}[t]{0.48\linewidth}
    \centering
    \caption{Synthetic 10K$\times$10K.}

    \begin{tabular}{|p{0.39\linewidth}|c|c|}
        \hline
        \textbf{Strategy} & \textbf{Time(s)} & \textbf{$\Delta$Acc.} \\
        \hline
        Random & 0.0073 & 0 \\
        \hline
        Rank-Inc. & 33.7 & 0.00017\% \\
        \hline
        Similarity & 9.86 & 0.0008\% \\
        \hline
    \end{tabular}
    
\end{subtable}
\hfill
\begin{subtable}[t]{0.48\linewidth}
    \centering
    \caption{Netflix Prize.}

    \begin{tabular}{|p{0.39\linewidth}|c|c|}
        \hline
        \textbf{Strategy} & \textbf{Trial(s)} & \textbf{Total(s)} \\
        \hline
        Random & 0.002 & 0.007 \\
        \hline
        Rank-Inc. & 50.34 & 151 \\
        \hline
        Similarity & 39.93 & 119.8 \\
        \hline
    \end{tabular}
    
\end{subtable}
\end{table}

Stage~I is the compute-intensive core of \projectname: it repeatedly applies
regularized matrix factorization (RMF) to the sampled submatrix \(R^*\) until
convergence. This module implements the \textit{Template} phase of the
Template–Extension execution model in Section~\ref{sec:design}.
From a systems perspective, the main difficulty is designing a synchronization engine that enforces these updates on GPUs without global write conflicts or excessive atomic operations while keeping Tensor Cores highly utilized.

\noindent\textbf{SGD Updates for RMF.}
We adopt the standard RMF objective in Equation~\ref{eq:rmf-obj} and optimize
it with stochastic gradient descent. For an observed entry \((j,i)\in\Omega\),
the per-entry gradients are:
\begin{equation} \label{eqn-update-p}
    \text{for } p_j:\quad
    \nabla f_{p_j}(p_j, q_i)
    = -2 \bigl(R_{ji} - p_j q_i^{\top}\bigr) q_i + 2\lambda p_j,
\end{equation}
\begin{equation} \label{eqn-update-q}
    \text{for } q_i:\quad
    \nabla f_{q_i}(p_j, q_i)
    = -2 \bigl(R_{ji} - p_j q_i^{\top}\bigr) p_j + 2\lambda q_i.
\end{equation}
In our multi-view extension, the item-factor update incorporates the similarity
regularizer from Equation~\ref{eqn-multi-view}, yielding
\begin{equation*}
\label{eqn-mvgmc-update}
\nabla f_{q_i}(p_j, q_i)
= -2\,(R_{ji} - p_j^{\top} q_i)\, p_j
  + 2\beta \sum_{i'=1}^n \bigl( S_{ii'}\, q_i - S_{ii'}\, q_{i'} \bigr),
\end{equation*}
where \(\sum_{i'} S_{ii'} (q_i - q_{i'})\) enforces similarity consistency
among items. These formulas are standard from an optimization viewpoint; rest of this subsection focuses on how \projectname turns them into a
high-throughput, conflict-aware GPU kernel.

\noindent\textbf{Sampling the TSID Template Submatrix.}
Stage~I does not operate on the full matrix. Instead, it only sees a
submatrix \(R^*\) obtained by sampling columns (or rows) from the original \(R\):
\(
R^* = \{ R_{ij} \mid j \in \mathcal{C},\, (i,j)\in\Omega \},
\)
where \(\mathcal{C}\) is the sampled column set with $|\mathcal{C}| = d = \lceil\alpha n\rceil$, and $k$ is the latent rank with $k \ll d \leq n$.
Under the TSID view, a well-sampled template provides the low-rank anchor
used by Stage~II, so the sampling strategy must be both cheap and robust.
We therefore use simple random column sampling. Table~\ref{tab:time} compares this choice with two heavier alternatives, \emph{rank increment} and \emph{similarity-based selection}~\cite{xie2019active}. Random sampling has negligible overhead on both workloads, while the heavier strategies are orders of magnitude slower and improve accuracy by less than \(10^{-3}\) on the synthetic workload. Thus, random sampling keeps the Stage~I startup cost negligible without measurable accuracy loss in our setting.

\noindent\textbf{Conflict-Free Parallel SGD.}
Naively parallelizing Equations~\ref{eqn-update-p}--\ref{eqn-update-q} on GPUs
creates severe write conflicts on shared item factors \(q_i\): many threads
process different users \(j\) that rated the same item \(i\), and all attempt
to update \(q_i\) concurrently.
Using global-memory atomics (e.g., \texttt{atomicAdd}) for every gradient
component serializes these updates and saturates the memory controller;
ignoring conflicts, in turn, harms convergence on power-law data with ``hot''
items.

\projectname resolves this by turning the per-entry updates into a
\emph{hierarchical reduction} over the GPU memory hierarchy.
Each thread first computes a local gradient contribution for its own copy of
\(q_i\) in registers.
Within a warp, the local contributions that target the same \(i\) are combined
using shuffle instructions, producing one partial sum per warp.
Warp partial sums are then accumulated in shared memory to obtain a
block-level delta for each \(q_i\).
Only this aggregated delta is finally applied to global memory, and we use at
most one atomic operation per block to commit it. Conflicts are absorbed at the register and shared-memory
levels, and global memory only sees a small number of coarse-grained,
commutative updates.

Formally, letting \(\nabla_{q_i} f(p_j,q_i)\) denote the per-thread
gradient, the aggregated update is
$q_{i}' = \frac{\sum_{j=1}^{k}
      \bigl(q_{i} + \eta \,\nabla_{q_i} f(p_j, q_i)\bigr)}{n} 
            = q_i + \eta\,\frac{\sum_{j=1}^{k}\nabla_{q_i} f(p_j, q_i)}{n} .$
The corresponding parallel update rules are: 
\begin{equation}\label{eqn-9}
p_{j}' \;=\; p_{j}\;+\;\eta \,\nabla_{p_j} f\bigl(p_j,q_i\bigr) ,~
q_{i}' \;=\; q_{i}\;+\;\eta \,\frac{\nabla_{q_i} f\bigl(p_j,q_i\bigr)}{n} ,
\end{equation}
where \(1/n\) factor normalizes the step size for the number of
contributing threads.
We refer to this scheme as \emph{conflict-free hierarchical synchronization}:
fine-grained, high-contention updates never directly contend in global memory,
and the remaining global atomics operate only on aggregated deltas.

\begin{figure}[t]
\centering
\includegraphics[width=0.86\columnwidth]{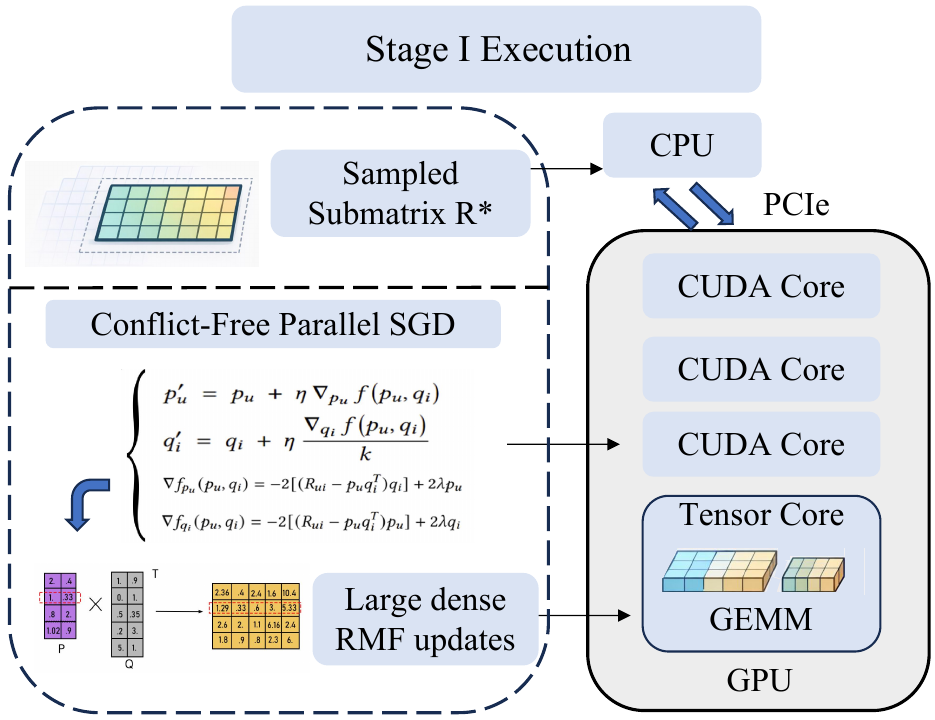}
\caption{Stage~I as a conflict-free engine: sparse RMF updates are aggregated and mapped onto Tensor Core tiles.}
\label{stage1structure}
\end{figure}

\noindent\textbf{Regularizing Sparse RMF Updates for Tensor Cores.}
The SGD updates in RMF are \emph{irregular} in the sense that each rating
touches a different pair of user and item factors, so threads access
non-contiguous rows of \(P\) and \(Q\) and issue fine-grained scatter/gather
operations instead of reading or writing contiguous tiles.
Even with conflict-aware aggregation, such sparse, index-based updates would
keep the kernel memory-bound unless we can reshape them into dense matrix
operations that fully utilize Tensor Cores.
\projectname therefore batches RMF updates into small matrix–matrix
multiplications over the low-rank factors rather than over the sparse rating
matrix itself.
In each SGD iteration, we group a mini-batch of users and items into dense
tiles and express the dominant computations as matrix multiplications
\(\Delta J = \mathrm{MatMul}(P, Q)\).
Tensor Cores perform FP16 multiplications with FP32 accumulation, providing
high throughput without sacrificing accuracy.
We implement this kernel in PyTorch using Automatic Mixed Precision (AMP),
which manages loss scaling and precision transitions, while the sparse access
pattern is confined to lightweight gather/scatter around the GEMMs.
Figure~\ref{stage1structure} illustrates the Stage~I workflow as a GPU-resident aggregation-and-compute engine: sparse RMF accesses are handled by efficient sparse-matrix primitives, while the dominant factor computations are mapped to dense Tensor Core kernels for \(\Delta J=\mathrm{MatMul}(P,Q)\).

\noindent\textbf{Stage~I Execution Example.}
Given a sparse input matrix \(R \in \mathbb{R}^{n_1 \times n_2}\),
Stage~I first samples subcolumns according to \(\alpha\) to form
\(R^* \in \mathbb{R}^{n_1 \times d}\) with \(d = \alpha n_2\).
We then initialize \(P \in \mathbb{R}^{n_1 \times k}\) and
\(Q \in \mathbb{R}^{d \times k}\) and run the GPU-parallel SGD kernel based on
\eqref{eqn-9} until \(R^* \approx P Q^{\top}\) converges.
This kernel is fully GPU-resident and is the main consumer of GPU FLOPs in Stage~I; hierarchical aggregation keeps global memory traffic and atomic operations low, allowing Tensor Cores to run in a compute-bound regime.

\subsection{Roofline-Driven Pipelined Dataflow}
\label{sec:opt2}

Once Stage~I has recovered the TSID template $\widehat{R^*}$ and its factors, Stage~II must reconstruct the full matrix under tight GPU-memory constraints. From a systems perspective, Stage~II is a \emph{pipelined dataflow engine} whose goal is to stream chunks through the GPU at a rate determined by a roofline-style balance between computation and PCIe bandwidth.

\noindent\textbf{Chunk-Based ALS Kernel.}
We use a chunked ALS-style reconstruction.
The full matrix \( \widehat{R} \in \mathbb{R}^{m\times n} \) is recovered from a restored submatrix \( \widehat{R^*} \in \mathbb{R}^{m\times \alpha n} \) (the template) using the objective in Equation~\eqref{eqn-10}:
\begin{equation}\label{eqn-10}
\mathcal{L}(C)
= \sum_{(i,j)\in\Omega}\Bigl(R_{ij} - \widehat{R^*}_i\, C_j\Bigr)^2
+\lambda\|C\|_F^2,
\end{equation}
where \(C\in\mathbb{R}^{\alpha n\times n}\) is the extension coefficient matrix whose \(j\)-th column \(C_j\) represents column \(j\) in the fixed template basis \(\widehat{R^*}\), and \(\Omega\) indexes observed entries.
We update \(C\) using gradient descent with the Adam optimizer~\cite{kingma2014adam}:
\(
C \leftarrow C -\eta \cdot \mathrm{Adam}\Bigl(\nabla_C \mathcal{L}(C)\Bigr),
\)
where \(\eta\) is the learning rate.

To avoid materializing all columns of \(C\) on the GPU, we partition \(R\) (and thus \(C\)) into column chunks of width \(\Delta c\), and process each chunk independently given the fixed template.




\noindent\textbf{Double-Buffered Pipelining and Iteration Cost.}
We implement Stage~II as a double-buffered pipeline that overlaps GPU computation with data transfers (see Figure~\ref{stage II}). While the GPU processes the ALS gradient step on chunk \(i\), chunk \(i+1\) is asynchronously prefetched via PCIe, and the results of chunk \(i-1\) are written back to CPU memory.

\noindent\textbf{Semantic Memory Management.}
Chunking provides a tunable memory footprint, which allows
\projectname to operate on matrices that are far larger than GPU memory.
The full matrix and $C$ reside in CPU memory, and only one or two chunks are
resident on the GPU at a time. The runtime uses multiple CUDA streams to
overlap three activities: (1) host-to-device prefetch of the next chunk
into a pre-allocated device buffer, (2) ALS computation on the current chunk, and (3) device-to-host flush of the previous chunk into a ring buffer allocated in pinned host memory. In contrast, the template $\widehat{R^*}$ and its factors $(P,Q)$ are read-mostly and reused across all chunks, so they are pinned in GPU HBM in Stage~II. This \emph{semantic} split, e.g., template vs. streaming chunks, treats different data structures according to their access patterns rather than managing all memory uniformly.

\noindent\textbf{Tensor Core Acceleration in Stage~II.}
The ALS update repeatedly multiplies $\widehat{R^*}$ with column blocks of $C$.
We express these operations as batched matrix multiplications that run on
Tensor Cores in mixed precision.  PyTorch AMP handles FP16/FP32
transitions and dynamic loss scaling.  Compared to a pure CUDA-core
implementation, this reduces the ALS compute time by up to an order of
magnitude on our largest datasets without loss in reconstruction accuracy.

\noindent\textbf{Stage~II Execution Example.}
As shown in Figure~\ref{stage II}, after Stage~I the submatrix \(R^*\) is fixed and serves as the template for Stage~II. To reconstruct the full matrix, we partition the original matrix into blocks that fit in GPU memory, then load each block onto the GPU in turn. For each block, we run ALS with \(R^*\) fixed while updating the factor matrix \(C \in \mathbb{R}^{\alpha n\times n}\) as in Equation~\ref{eqn-10}, repeating until target accuracy or a maximum iteration count is reached. Once a block is completed, it is offloaded to CPU memory and the next block starts, as illustrated in Figure~\ref{stage II}. 

\begin{figure}[t]
\centering
\includegraphics[scale=0.15]{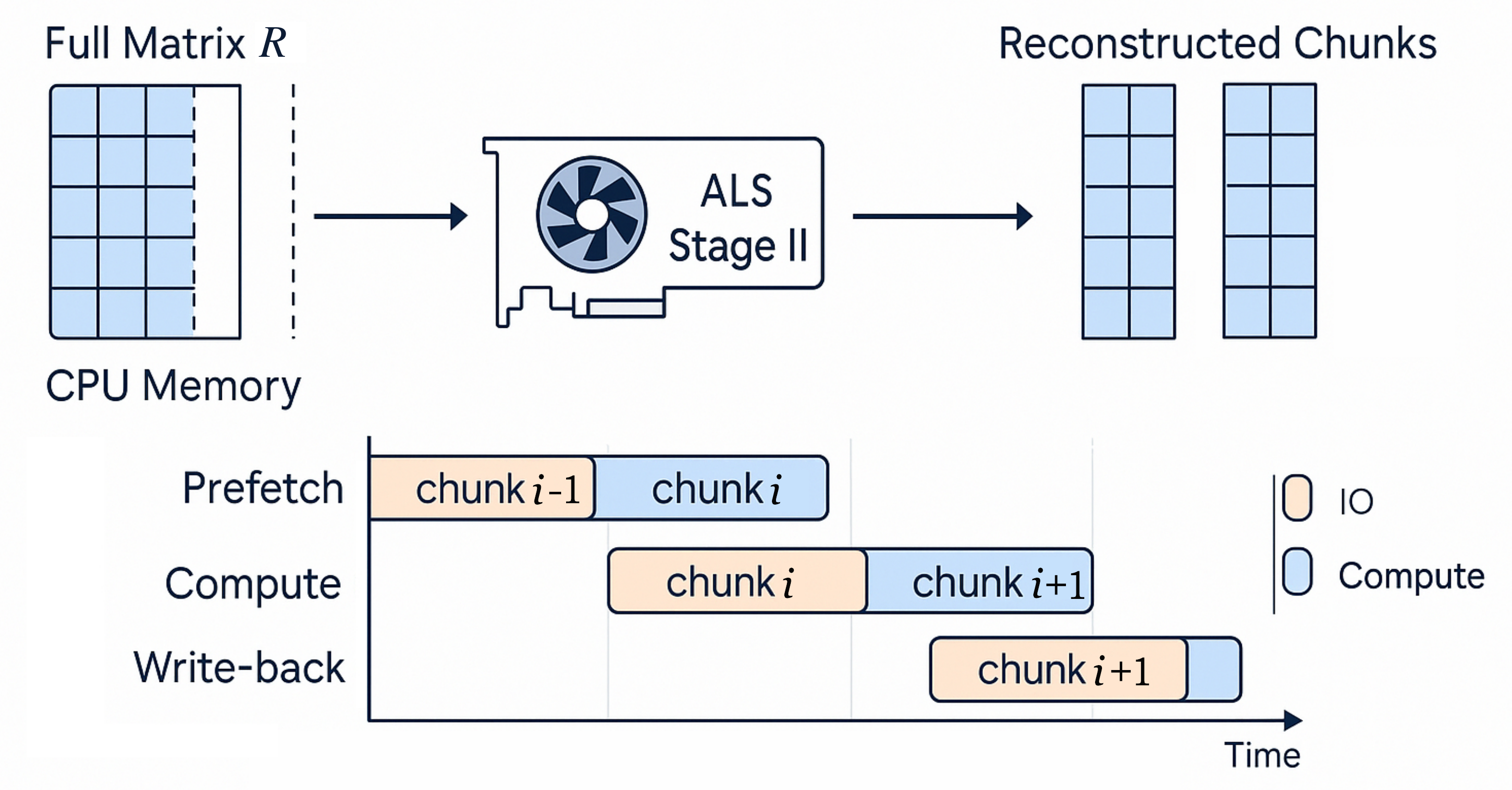}
\caption{Stage~II as a roofline-driven pipelined dataflow engine: ALS compute on chunk $i$ overlaps with prefetch of $i+1$ and flush of $i-1$.}
\label{stage II}
\end{figure}

\subsection{Asymmetric Gradient Guardrail}
\label{sec:opt3}

Although SGD-based RMF delivers high throughput, its Stage~I updates can become numerically unstable or even diverge when gradients grow large and unbounded, especially under FP16 execution on Tensor Cores. Instead of treating this purely as an ``optimization trick'', \projectname integrates gradient clipping as a \emph{numerical stabilization guardrail} that enables aggressive mixed-precision computation without sacrificing the convergence of the whole algorithm.

\noindent\textbf{Asymmetric Gradient Clipping.}
To ensure stable convergence and enhance accuracy, we refine the SGD updates in Equation~\ref{eqn-9} by rewriting them as:
\begin{equation*}\label{eqn-9-again}
\begin{aligned}
p_{j}' \;&=\; p_{j} \;+\; \eta \,\min \Bigl(\max \bigl(\xi,\,-\nabla_{p_j} f(p_j,q_i)\bigr),\,\nabla_{p_j} f(p_j,q_i)\Bigr),\\
q_{i}' \;&=\; q_{i} \;+\; \frac{\eta }{n}\,\min \Bigl(\max \bigl(\xi,\,-\nabla_{q_i} f(p_j,q_i)\bigr),\,\nabla_{q_i} f(p_j,q_i)\Bigr),
\end{aligned}
\end{equation*}
where \(\xi > 0\) is the clipping threshold. For each coordinate, positive gradients larger than \(\xi\) are truncated to \(\xi\), while negative gradients are left unchanged. This \emph{asymmetric} clipping lets parameters shrink freely but caps explosive growth.

Formally, if $\nabla f > \xi$, the update is forced to use $\xi$; multiplying by the step size \(\eta\) then guarantees
\(
  \|\Delta p\|_\infty \le \eta\,\xi
\)
per coordinate, preventing exploding updates and numerical overflow. In FP16, where the representable range is limited, this bound is critical: it ensures that even in early epochs with a large reconstruction error, the RMF kernel stays within a safe dynamic range. Combined with AMP's loss scaling, this guardrail allows \projectname to double effective memory bandwidth and compute throughput compared to a pure FP32 implementation while maintaining stable convergence.

\noindent\textbf{Effect on Convergence.}
Figure~\ref{fig:grad_clip_discrete} compares the loss over iterations with
and without clipping.  Without clipping (blue), the loss oscillates and
fails to converge, as large gradients repeatedly overshoot the optimum.
With clipping (red), the loss descends smoothly and stabilizes near the
minimum.  In our end-to-end experiments, clipping consistently improves
accuracy and reduces the number of iterations needed to reach a target
reconstruction error, with negligible runtime overhead. From a system viewpoint, this means the mixed-precision engine can be run in a ``fail-safe'' mode: even aggressive TSID sampling and high learning rates do not cause catastrophic divergence.

\begin{figure}[t]
  \centering
  \begin{tikzpicture}[scale=1.0]
    \begin{axis}[
        name=unclipped,
        width=0.5\columnwidth,
        height=0.35\columnwidth,
        title={Without Clipping},
        xlabel={Iteration},
        ylabel={Loss},
        grid=both,
        font=\scriptsize,
        ymin=0.4, ymax=1.0,
        xmin=0, xmax=100,
      ]
      \addplot[
        mark=*, mark size=1.5pt, thick, blue
      ] coordinates {
        (0,0.80) (10,0.75) (20,0.85) (30,0.70)
        (40,0.90) (50,0.65) (60,0.95) (70,0.60)
        (80,0.92) (90,0.62) (100,0.88)
      };
    \end{axis}

    \begin{axis}[
        at={(unclipped.east)},
        anchor=west,
        xshift=1cm,
        width=0.5\columnwidth,
        height=0.35\columnwidth,
        font=\scriptsize,
        title={With Clipping},
        xlabel={Iteration},
        grid=both,
        ymin=0, ymax=1.0,
        xmin=0, xmax=100,
      ]
      \addplot[
        mark=*, mark size=1.5pt, thick, red
      ] coordinates {
        (0,0.80) (10,0.70) (20,0.60) (30,0.55)
        (40,0.50) (50,0.48) (60,0.47) (70,0.46)
        (80,0.45) (90,0.44) (100,0.44)
      };
    \end{axis}
  \end{tikzpicture}
  \caption{
Comparison of loss over iterations with and without gradient clipping (threshold \( \xi \)).
}
\label{fig:grad_clip_discrete}
\end{figure}
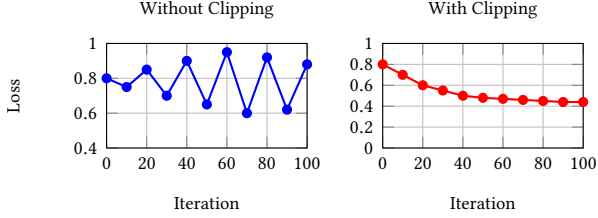

\section{Evaluation}~\label{sec:eva}

\stitle{Implementation Details.} \projectname\ is implemented in about 7{,}000 lines of Python. Using PyTorch’s device management, inputs are automatically moved to the target device, so all tensor operations in parallel SGD and ALS routines run on GPU. The code is organized as a modular, object-oriented framework that encapsulates submatrix sampling, iterative factorization, and gradient clipping, with logging and visualization utilities for debugging and performance analysis. For Tensor Core acceleration, we use PyTorch~\cite{paszke2019pytorch} with CUDA~12.4~\cite{wilt2013cuda}, and run critical batched matrix multiplications in mixed precision to fully exploit NVIDIA Tensor Cores.
%

\subsection{Experimental Setup}

\textbf{Test Platform.} Our experiments are conducted on a server equipped with two Intel(R) Xeon(R) Gold 6426Y processors with 512GB of CPU RAM and 32 threads, running Ubuntu 20.04.6 LTS with CUDA 12.4. All methods that require a graphics processing unit are executed on four NVIDIA A100 GPUs, each equipped with 40GB of HBM.

\noindent\textbf{Datasets.}
We evaluate on both synthetic and real-world datasets (Table~\ref{table-dataset}). Synthetic datasets emulate the high sparsity of real recommenders, with 80\% of user--item entries missing. We generate each synthetic matrix as the product of two randomly sampled low-rank factor matrices, \(P Q^\top\), and then randomly mask entries to the target sparsity, ensuring that the underlying complete matrix is low-rank. They range from \(1\text{K} \times 1\text{K}\) to \(100\text{K} \times 150\text{K}\) and are used to study behavior under varying sparsity and matrix sizes.
For real-world workloads, we include both image and recommendation datasets. Image Construction and Noise Clear support an image reconstruction and denoising case study, with qualitative results shown in Figure~\ref{noise}; both are derived from grayscale bridge images from the public PxHere repository, with pixels treated as matrix entries~\cite{krajewska2024matrix}. Last.fm 2K is a music recommendation dataset containing user--artist interactions and auxiliary tag information~\cite{cantador2011hetrec}. MovieLens 25M and MovieLens Latest 33M provide large-scale user--movie rating matrices from the GroupLens benchmark suite~\cite{harper2015movielens}. Netflix Prize contains 100{,}480{,}507 ratings from 480{,}189 users on 17{,}770 movies~\cite{NetflixPrize}. Yahoo! Music, used in the 2011 KDD Cup, has 262{,}810{,}175 ratings from 200{,}000 users on 136{,}000 music items and is a standard benchmark for large-scale recommendation~\cite{yahoo2023dataset}.

\noindent\textbf{Data Model.}
Input matrices are represented in the coordinate (COO) format~\cite{chou2018format}, which stores only the row indices, column indices, and values of non-zero elements. This representation is space-efficient for sparse matrices.

\begin{table}[t]
	\centering
        \caption{Datasets analysis.}
    \label{table-dataset}
    \resizebox{0.99\columnwidth}{!}{ 
	
	\begin{tabular}{|c|c|c|c|c|}
		\hline
        {\begin{tabular}[c]{@{}c@{}}Real-World\\ Dataset\end{tabular}}& \textbf{Users} & \textbf{Items} & \textbf{Ratings} & \textbf{Size}\\
		\hline
        Image Construction~\cite{krajewska2024matrix} & 1080 & 1920 & - & 3.2MB\\
		\hline
        Noise Clear~\cite{krajewska2024matrix} & 771 & 621 & - & 44KB\\
        \hline
        Last.fm 2K~\cite{cantador2011hetrec} & 1.9K & 17K & -- & 123MB\\
		\hline        MovieLens 25M~\cite{harper2015movielens} & 162K & 59K & 25M & 35.6GB\\
        \hline
        MovieLens Latest 33M~\cite{harper2015movielens} & 331K & 83K & 33M & \textbf{102.4GB}\\
        \hline
        Netflix Prize~\cite{NetflixPrize} & 480,189 & 17,770 & 100,480,507 & 31GB \\
		\hline
		Yahoo! Music~\cite{yahoo2023dataset} & 200,000 & 136,000 & 262,810,175 & {101GB}\\
		\hline
        \hline
         \multirow{9}{*}{\begin{tabular}[c]{@{}c@{}}Synthetic\\ Dataset\end{tabular}}  & \textbf{Row} & \textbf{Column} & \textbf{Sparsity} & \textbf{Size}\\
        \cline{2-5}
       
        & 1K & 1K & \multirow{8}{*}{80\%} & 3.8MB\\
        \cline{2-3} \cline{5-5}
        & 1K & 3K & & 11MB\\
        \cline{2-3} \cline{5-5}
        & 3K & 3K & & 34MB\\
        \cline{2-3} \cline{5-5}
        & 3K & 10K & & 114MB\\
        \cline{2-3} \cline{5-5}
        & 10K & 100K & & 3.72GB \\
        \cline{2-3} \cline{5-5}
        & 50K & 100K & & 18.6GB\\
        \cline{2-3} \cline{5-5}
        & 100K & 100K & & 37.2GB\\
        \cline{2-3} \cline{5-5}
        & 100K & 150K & & 55.87GB\\
        \hline
	\end{tabular}
    }

\end{table}

\noindent\textbf{Baselines and Variants.}
We compare \projectname against five external state-of-the-art (SOTA) matrix-completion solvers
and several internal variants.
For NN-centric methods, we use \textbf{CSNN} as the nuclear-norm SOTA baseline.
For SVD-centric, kernel-based, and graph-based methods, we use \textbf{CSPGD}, \textbf{CGM},
\textbf{NTK}, and \textbf{MVGMC} as SOTA baselines on their respective tasks.
Implementations of CSNN, CSPGD, and CGM follow the open-source code from~\cite{krajewska2024matrix};
NTK uses the implementation from~\cite{radhakrishnan2022simple};
MVGMC follows~\cite{koohi2019parallel}.
These five methods appear in Table~\ref{over} as the \emph{external SOTA baselines}.

\noindent\textbf{Accuracy Evaluation Metrics.}
To measure accuracy, we adopt the normalized mean absolute error (NMAE)~\cite{chicco2021coefficient}, which is widely used in matrix completion studies~\cite{wen2012lmafit,todeschini2013probabilistic,krajewska2024matrix,xiong2022traffic}. Given a held-out test set of observed entries $\Omega_{\text{test}}$, NMAE is defined as
$\frac{1}{|\Omega_{\text{test}}|(m_{\max}-m_{\min})}\sum_{(i,j)\in\Omega_{\text{test}}} |\widehat{R}_{ij}-R_{ij}|$,
where $m_{\max}$ and $m_{\min}$ are the maximum and minimum values in $R$. Lower NMAE indicates higher reconstruction accuracy.

\noindent\textbf{Parameter Selection.}
We typically set the subcolumn selection ratio (i.e., the percentage of total columns selected for each submatrix) to 0.5 for small datasets to ensure high accuracy and to 0.1 for large-scale datasets to balance memory usage and computational efficiency. The number of chunks is generally chosen to be between 10 and 30 to achieve optimal performance (detailed in Section~\ref{eva:sen}).

\begin{table*}[t]
\centering
\caption{End-to-end performance comparison in terms of accuracy (NMAE) and runtime (seconds). The symbol \textcolor{red}{\ding{55}} denotes unavailable results due to excessive runtime or missing implementations. NTK requires a strict training setup, making it unavailable on most datasets. We set $\alpha = 0.1$ for all experiments below.}

\label{over}
\resizebox{\textwidth}{!}{
\begin{tabular}{|l|c|c|c|c|c|c|c|c|c|c|c|c|}
\hline
\multicolumn{1}{|c|}{\multirow{2}{*}{\begin{tabular}[c]{@{}c@{}}Dataset\end{tabular}}} & \multicolumn{2}{c|}{CSNN} & \multicolumn{2}{c|}{CSPGD} & \multicolumn{2}{c|}{CGM} & \multicolumn{2}{c|}{MVGMC} & \multicolumn{2}{c|}{NTK} & \multicolumn{2}{c|}{\textbf{\projectname (Ours)}} \\
\cline{2-3} \cline{4-5} \cline{6-7} \cline{8-9} \cline{10-11} \cline{12-13}
& NMAE & Time (s) & NMAE & Time (s) & NMAE & Time (s) & NMAE & Time (s) & NMAE & Time (s) & NMAE & Time (s) \\
\hline
1K$\times$3K & \colorbox{green}{\textbf{1.08 × $\textbf{10}^{\textbf{-6}}$}} & \colorbox{pink}{{21.40K} (\textbf{8136x})} & 0.0014 & 51.00 (19.3x) & 0.0028 & 532.00 (202x) & \colorbox{pink}{0.3662} & 7.98 (3.03x)& \textcolor{red}{\ding{55}} & \textcolor{red}{\ding{55}} & 0.0008 & \colorbox{green}{\textbf{2.63}} \\
\hline
3K$\times$3K & \colorbox{green}{{\textbf{0.0013}}} & \colorbox{pink}{{140.00K} (\textbf{11647x})} & 0.0029 & 116.00 (9.65x) & \colorbox{pink}{0.8419} & 487.00 (40.5x) & 0.3953 & 34.55(2.87x) & \textcolor{red}{\ding{55}} & \textcolor{red}{\ding{55}} & 0.0027 & \colorbox{green}{\textbf{12.02}} \\
\hline
3K$\times$10K & \textbf{0.0020} & \colorbox{pink}{{260.00K} (\textbf{10400x})} & 0.0020 & 145.00 (5.8x) & \colorbox{pink}{0.8438} & 579.00 (23.1x) & 0.3450 & 141.93 (5.67x) & \textcolor{red}{\ding{55}} & \textcolor{red}{\ding{55}} & \colorbox{green}{\textbf{0.0019}} & \colorbox{green}{\textbf{25.00}} \\
\hline
Image Construction & 0.0668 & {4.62K} (872x) & 0.0474 & \colorbox{pink}{{7.30K} (\textbf{1377x})} & 0.0672 & 90.00 (17x) & \colorbox{pink}{0.2447} & 19.23 (3.62x) & 0.1100 & 21.00 (3.96x) & \colorbox{green}{\textbf{0.0454}} & \colorbox{green}{\textbf{5.31}} \\
\hline
Last.fm 2K & 0.4134 & \colorbox{pink}{{495.40} (\textbf{643x})} & 0.2281 & 43.90 & 0.1899 & 1.24 (17x) & 0.1900 & 3.52 (3.62x) & \colorbox{pink}{0.4530} & 0.85 (3.96x) & \colorbox{green}{\textbf{0.1496}} & \colorbox{green}{\textbf{0.77}} \\
\hline
MovieLens 25M & \textcolor{red}{\textbf{OOM}} & \textcolor{red}{\textbf{OOM}} & \textcolor{red}{\textbf{OOM}} & \textcolor{red}{\textbf{OOM}} & \textcolor{red}{\textbf{OOM}} & \textcolor{red}{\textbf{OOM}} & \textcolor{red}{\textbf{OOM}} & \textcolor{red}{\textbf{OOM}} & \textcolor{red}{\ding{55}} & \textcolor{red}{\ding{55}} & \colorbox{green}{\textbf{0.1352}} & \colorbox{green}{\textbf{2.88K}} \\
\hline
MovieLens Latest 33M & \textcolor{red}{\textbf{OOM}} & \textcolor{red}{\textbf{OOM}} & \textcolor{red}{\textbf{OOM}} & \textcolor{red}{\textbf{OOM}} & \textcolor{red}{\textbf{OOM}} & \textcolor{red}{\textbf{OOM}} & \textcolor{red}{\textbf{OOM}} & \textcolor{red}{\textbf{OOM}} & \textcolor{red}{\ding{55}} & \textcolor{red}{\ding{55}} & \colorbox{green}{\textbf{0.1385}} & \colorbox{green}{\textbf{5.04K}} \\
\hline
Netflix Prize & \textcolor{red}{\textbf{OOM}} & \textcolor{red}{\textbf{OOM}} & \textcolor{red}{\textbf{OOM}} & \textcolor{red}{\textbf{OOM}} & \textcolor{red}{\textbf{OOM}} & \textcolor{red}{\textbf{OOM}} & \textcolor{red}{\textbf{OOM}} & \textcolor{red}{\textbf{OOM}} & \textcolor{red}{\ding{55}} & \textcolor{red}{\ding{55}} & \colorbox{green}{\textbf{0.1382}} & \colorbox{green}{\textbf{53.91}} \\
\hline
Yahoo! Music & \textcolor{red}{\textbf{OOM}} & \textcolor{red}{\textbf{OOM}} &\textcolor{red}{\textbf{OOM}} & \textcolor{red}{\textbf{OOM}} & \textcolor{red}{\textbf{OOM}} & \textcolor{red}{\textbf{OOM}} & \textcolor{red}{\textbf{OOM}} & \textcolor{red}{\textbf{OOM}} & \textcolor{red}{\ding{55}} & \textcolor{red}{\ding{55}} & \colorbox{green}{\textbf{0.1473}} & \colorbox{green}{\textbf{1.01K}} \\
\hline
\end{tabular}
}
\end{table*}


\subsection{Overall Performance}

Table~\ref{over} compares \projectname\ with five SOTA baselines (CSNN, CSPGD, CGM, MVGMC, and NTK) on diverse matrix completion tasks using NMAE and end-to-end runtime.
\projectname\ consistently achieves the best overall efficiency while maintaining strong accuracy. On mid-sized matrices, it delivers up to \textbf{11,647$\times$} speedup over NN-centric methods with comparable accuracy. Against SVD-centric methods, it achieves up to \textbf{1,377$\times$} speedup and up to \textbf{99.7\%} lower NMAE.
For accuracy, on a 1K×3K matrix, \projectname\ achieves an NMAE of 0.00075, improving over CSPGD and CGM by \textbf{46.4\%} and \textbf{73.1\%}, respectively. On a 3K×3K matrix, it attains an NMAE of 0.0027, nearly matching CSNN and outperforming CGM by \textbf{99.7\%}. On image reconstruction, it achieves an NMAE of 0.04543, improving over CSNN and CGM by \textbf{31.9\%} and \textbf{32.4\%}. These gains mainly come from gradient clipping, which stabilizes training and reduces reconstruction error.

\noindent\textbf{Real-World Large Dataset Testing.}
On MovieLens (25M and Latest 33M), Netflix and Yahoo! Music datasets, CSNN, CSPGD, CGM, and MVGMC fail with out-of-memory errors due to full-matrix execution. NTK is not reported because its required training setup is unavailable for these datasets. In contrast, \projectname\ completes Netflix with NMAE 0.1382 in 53.91s and Yahoo! Music with NMAE~0.1473 in 1.01Ks. Although Yahoo! Music and MovieLens Latest 33M have comparable storage sizes, their runtimes differ substantially, indicating that end-to-end time is affected not only by data size but also by matrix shape, sparsity pattern, and chunk-level reconstruction cost. On Last.fm 2K, \projectname\ achieves the best NMAE of 0.1496 and finishes in 0.77s, outperforming the slowest highlighted baseline CSNN by \textbf{643$\times$}. On MovieLens 25M and MovieLens Latest 33M, all baselines run out of memory or are unavailable, while \projectname\ completes them with NMAE~0.1352 and 0.1385, respectively. Entries marked with \ding{55} in Table~\ref{over} denote such failures. Although NTK attains low NMAE on some small image tasks, its scalability is limited. Overall, \projectname\ combines strong accuracy with substantially better efficiency and scalability.

\noindent\textbf{Noise Clear.}
To complement the quantitative results, we evaluate \projectname\ on a real-world image denoising task. We compare against CSPGD~\cite{krajewska2024matrix}, the strongest baseline in Table~\ref{over}. Figure~\ref{noise} shows the original image and the reconstructions produced by CSPGD and \projectname. \projectname\ removes noise more effectively and produces a clearer reconstruction, providing evidence of improved reconstruction quality.

\begin{figure}[t]
\centering
\subfloat[Origin image.]
{
    \begin{minipage}[b]{0.25\linewidth}
        \centering
        \includegraphics[scale=0.2]{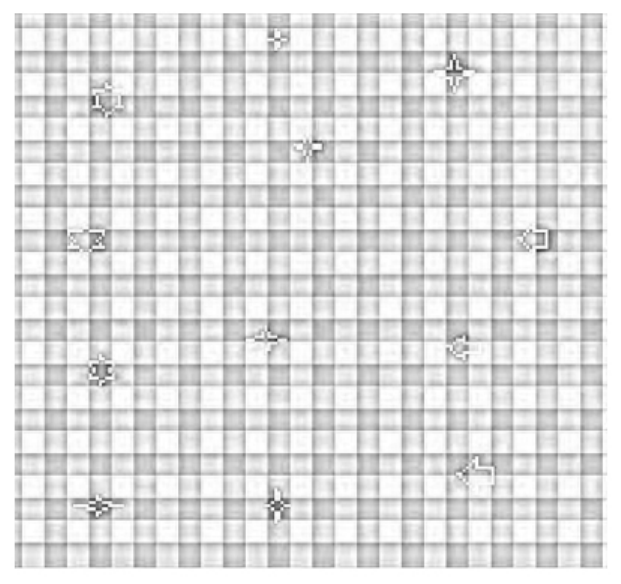}
    \end{minipage}
}
\subfloat[CSPGD.]{
    \begin{minipage}[b]{0.25\linewidth}
        \centering
        \includegraphics[scale=0.2]{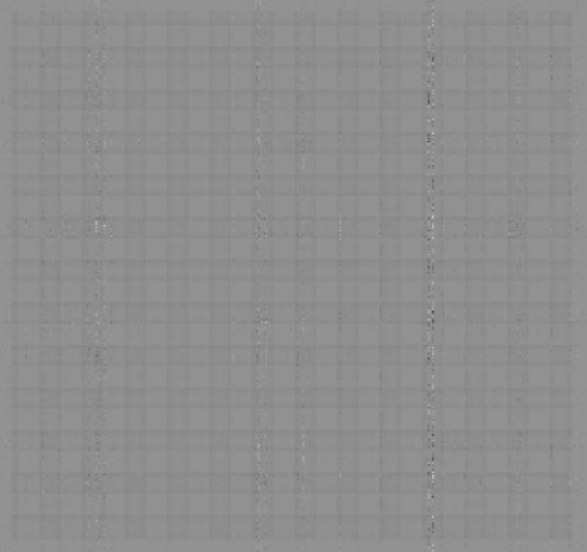}
    \end{minipage}
}
\subfloat[\projectname (Ours).]{
    \begin{minipage}[b]{0.25\linewidth}
        \centering
        \includegraphics[scale=0.2]{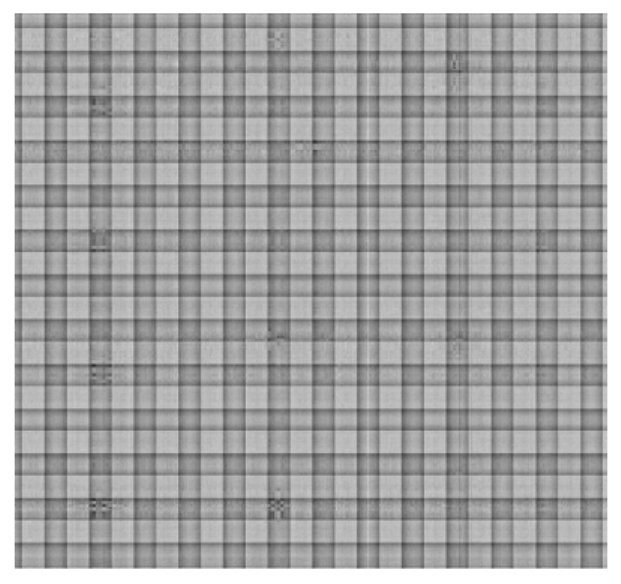}
    \end{minipage}
}
\caption{Qualitative comparison for image denoising.}
\label{noise}
\end{figure}

\begin{table}[t]
\centering
\caption{Profiled GPU memory usage in GB. \projectname uses the chunk-based pipeline; $>100$ indicates that the method exhausts the largest available GPU-memory budget (more than 100GB) and fails with OOM.}
\label{tab:memory_usage_combined}
\footnotesize
\begin{tabular}{|l|c|c|c|c|c|}
\hline
\textbf{Dataset} & \textbf{CSNN} & \textbf{CSPGD} & \textbf{CGM} & \textbf{MVGMC} & \textbf{\projectname} \\
\hline
1K $\times$ 3K & 0.84 & 0.80 & 0.83 & 0.76 & \textbf{0.30} \\
\hline
3K $\times$ 3K & 2.15 & 0.84 & 1.44 & 0.78 & \textbf{0.47} \\
\hline
3K $\times$ 10K & 6.74 & 0.90 & 4.95 & 3.22 & \textbf{0.79} \\
\hline
Last.fm 2K & 4.2 & 0.95 & 3.4 & 1.6 & \textbf{0.85} \\
\hline
MovieLens 25M & $>100$ & $>100$ & $>100$ & $>100$ & \textbf{27} \\
\hline
MovieLens 33M & $>100$ & $>100$ & $>100$ & $>100$ & \textbf{38} \\
\hline
Netflix Prize & $>100$ & $>100$ & $>100$ & $>100$ & \textbf{22} \\
\hline
Yahoo! Music & $>100$ & $>100$ & $>100$ & $>100$ & \textbf{23} \\
\hline
\end{tabular}
\end{table}

\subsection{Memory Usage Comparison}

Table~\ref{tab:memory_usage_combined} reports the profiled GPU memory usage on synthetic and real-world recommendation workloads. On synthetic matrices, \projectname keeps the resident footprint below 0.8GB, while baselines require up to 6.74GB. On the 3K $\times$ 10K, \projectname reduces memory usage from 6.74GB for CSNN and 4.95GB for CGM to only 0.79GB. On Last.fm 2K, \projectname uses 0.85GB, lower than all baselines. On larger real-world workloads, CSNN, CSPGD, CGM, and MVGMC exceed the largest available GPU-memory budget in our testbed and fail with OOM, whereas \projectname completes within 27GB on MovieLens 25M, 38GB on MovieLens Latest 33M, 22GB on Netflix, and 23GB on Yahoo! Music. Overall, the chunk-based pipeline keeps the resident working set within a single 40GB A100 GPU and enables larger-scale completion than full-matrix baselines.

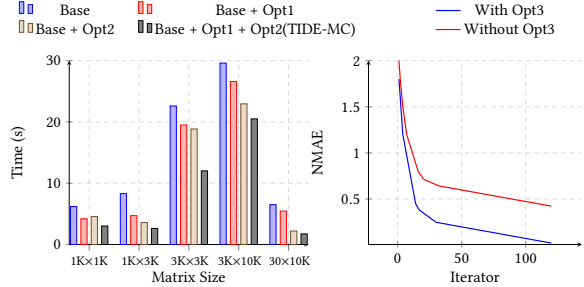
\begin{figure}[t]
  \centering
  \resizebox{0.9\columnwidth}{!}{%
    \begin{tikzpicture}[]
      \begin{axis}[
          scale = 0.8,
          at={(-0.4,0)},
          anchor=origin,
          width=0.4\textwidth,
          height=0.4\textwidth,
          font=\large,
          axis x line=bottom,
          axis y line=left,
          xlabel={Matrix Size},
          ylabel={Time (s)},
          xtick={1,2,3,4,5},
          xticklabels={1K$\times$1K,\,1K$\times$3K,\,3K$\times$3K,\,3K$\times$10K,\,30$\times$10K},
          ymin=0, ymax=30,
          grid=both,
          grid style={dashed,gray!30},
          legend style={
           font=\large,
            at={(0.5,1.1)},
            anchor=south,
            legend columns=2,
            /tikz/every even column/.append style={column sep=0.5cm},
            draw=none
          },
          ybar=0.1cm,   
          bar width=0.15cm,  
          x=1.5cm,    
          enlarge x limits=0.1,  
          yticklabel style={/pgf/number format/fixed},
          xticklabel style={font=\small},
        ]
        \addplot coordinates {(1,6.2) (2,8.30) (3,22.6) (4,29.6) (5,6.5)};
        \addplot coordinates {(1,4.20) (2,4.71) (3,19.5) (4,26.6) (5,5.46)};
        \addplot coordinates {(1,4.56) (2,3.58) (3,18.85) (4,22.95) (5,2.2)};
        \addplot coordinates {(1,3.02) (2,2.63) (3,12.02) (4,20.5) (5,1.74)};
        \legend{Base, Base + Opt1, Base + Opt2, Base + Opt1 + Opt2(\projectname)}
        
      \end{axis}
      
     \begin{axis}[
        scale=0.8,
        width=0.45\textwidth,
        height=0.4\textwidth,
        at={(7.2cm,0)},  
        axis x line=bottom,
        axis y line=left,
        xlabel={Iterator},
        ylabel={NMAE},
        ymin=0.007, ymax=2,
        grid=both,
        font=\large,
        grid style={dashed,gray!30},
        legend style={
            at={(0.6,1.1)},
            anchor=south,
            font=\large,
            legend columns=1,
            /tikz/every even column/.append style={column sep=0.5cm},
            draw=none
        },
        enlarge x limits=0.2,
        yticklabel style={/pgf/number format/fixed, /pgf/number format/precision=3}, 
    ]
        \addplot [color=blue, thick] coordinates {
            (1, 1.8)
            (2, 1.6)
            (3, 1.4)
            (4, 1.2)
            (14, 0.45)
            (16, 0.4)
            (17, 0.38)
            (30, 0.248)
            (31, 0.245)
            (32, 0.243)
            (38, 0.228)
            (41, 0.22)
            (42, 0.218)
            (47, 0.205)
            (48, 0.203)
            (49, 0.20)
            (54, 0.188)
            (55, 0.185)
            (66, 0.158)
            (67, 0.155)
            (68, 0.153)
            (69, 0.15)
            (74, 0.138)
            (80, 0.123)
            (81, 0.12)
            (82, 0.118)
            (83, 0.115)
            (84, 0.113)
            (91, 0.095)
            (92, 0.093)
            (101, 0.07)
            (102, 0.068)
            (103, 0.065)
            (104, 0.063)
            (105, 0.06)
            (118, 0.028)
            (119, 0.025)
            (120, 0.023)
        };
        \addplot [color=red, thick] coordinates {
            (1, 2.0)
            (2, 1.8)
            (3, 1.65)
            (4, 1.5)
            (5, 1.4)
            (6, 1.3)
            (7, 1.2)
            (16, 0.8)
            (17, 0.78)
            (18, 0.76)
            (19, 0.74)
            (20, 0.72)
            (21, 0.71)
            (32, 0.643)
            (33, 0.64)
            (34, 0.638)
            (35, 0.635)
            (36, 0.633)
            (37, 0.63)
            (38, 0.628)
            (46, 0.608)
            (47, 0.605)
            (48, 0.603)
            (49, 0.60)
            (69, 0.55)
            (70, 0.548)
            (71, 0.545)
            (72, 0.543)
            (73, 0.54)
            (74, 0.538)
            (90, 0.498)
            (91, 0.495)
            (92, 0.493)
            (93, 0.49)
            (98, 0.478)
            (99, 0.475)
            (100, 0.473)
            (101, 0.47)
            (102, 0.468)
            (110, 0.448)
            (111, 0.445)
            (112, 0.443)
            (113, 0.44)
            (117, 0.43)
            (118, 0.428)
            (119, 0.425)
            (120, 0.423)
        };
        \legend{With Opt3, Without Opt3}
    \end{axis}
    \end{tikzpicture}%
  }
  \caption{Ablation study of \projectname configurations: \texttt{Base} (no optimizations), \texttt{Opt1} (conflict-free synchronization RMF), \texttt{Opt2} (chunk‑based pipelining), and \texttt{Opt3} (gradient clipping).}
  \label{fig:combined_merged}
\end{figure}

\subsection{Efficiency of \projectname's Optimizations}
We conduct an ablation study (Figure~\ref{fig:combined_merged}) to evaluate each optimization in \projectname on matrix sizes from 1K$\times$1K to 3K$\times$10K, plus an additional wide-matrix microbenchmark of 30$\times$10K. Across all sizes, we compare the base RMF+ALS version and different optimization combinations in execution time (left) and NMAE (right).

\noindent\textbf{Efficiency of Opt1.} Opt1, the conflict-free synchronization optimization in Stage I (Section~\ref{sec:opt1}), consistently reduces execution time across all matrix sizes, with larger benefits on bigger matrices. For example, it cuts runtime on a $1\mathrm{K}\times3\mathrm{K}$ matrix from 8.3s to 4.71s, with up to a \textbf{43\%} reduction overall.

\noindent\textbf{Efficiency of Opt2.} Opt2, the roofline-driven pipelined dataflow in Stage II (Section~\ref{sec:opt2}), consistently outperforms the base version, reducing runtime on a $30\times10\mathrm{K}$ matrix from 6.5s to 2.2s (\textbf{65\%}). Its chunk-based parallelism is especially effective when rows are much fewer than columns, making it well suited to wide real-world matrices.

\noindent\textbf{Efficiency of Opt3.} Opt3, the asymmetric gradient guardrail in Section~\ref{sec:opt3}, improves convergence and accuracy. Unless otherwise specified, we set the clipping threshold to \(\xi=0.1\) in all Opt3 experiments. As shown in Figure~\ref{fig:combined_merged}, clipping reduces NMAE from 2.0 to 1.8 at the first iteration and from 0.423 to 0.023 at the 120th iteration (94\% relative improvement). It also reaches NMAE=0.1 in about 60 iterations versus 100 without clipping, giving a \textbf{40\%} speedup.

\noindent\textbf{Efficiency of Tensor Cores.}
Figure~\ref{fig:combined_two_plots} compares the runtime with and without Tensor Core acceleration. We disable Tensor Cores with \texttt{TF\_DISABLE\_CUBLAS\_} \texttt{TENSOR\_OP\_MATH=1}, forcing all kernels onto CUDA cores. Tensor Cores consistently run faster, with the gap widening for larger matrices: for a ($16\mathrm{K}\times16\mathrm{K}$) matrix, the runtime drops from 30.7s (CUDA cores) to 22.18s (Tensor Cores), a \textbf{27.7\%} reduction that highlights the advantage for matrix computations.

\begin{figure}[t]
  \centering
  \begin{tikzpicture}
    \begin{groupplot}[
        group style={
          group size=2 by 1,
          horizontal sep=0.55cm
        },
        width=0.40\columnwidth,
        height=0.28\columnwidth,
        scale only axis,
        ybar,
        ymin=0,
        enlarge x limits=0.16,
        axis line style={black!70, line width=0.5pt},
        tick style={black!70, line width=0.5pt},
        ymajorgrids=true,
        x tick label style={
  rotate=-15,
  yshift=2pt,
},
        grid style={dashed, gray!25},
        xlabel={Matrix Size},
        xlabel style={font=\scriptsize},
        tick label style={font=\tiny},
        xticklabel style={align=center, font=\tiny},
      ]

      \nextgroupplot[
        ylabel={Time (s)},
        ylabel style={font=\scriptsize, at={(0.15,0.5)}},
        xtick={1,2,3,4},
        xticklabels={
          16K$\times$16K,
          32K$\times$32K,
          64K$\times$64K,
          80K$\times$80K
        },
        ymax=300,
        ytick={0,100,200,300},
        legend style={
          draw=none,
          font=\tiny,
          at={(0.03,0.97)},
          anchor=north west,
          legend columns=1
        },
        title={(a) Tensor Core vs.\ CUDA Core.},
        title style={
          at={(0.5,-0.44)},
          anchor=north,
          font=\footnotesize
        },
      ]
      \addplot[
        draw=blue!70!black,
        fill=blue!25,
        bar width=4.5pt,
        bar shift=-2.7pt
      ] coordinates {
        (1,30.7) (2,62.16) (3,158.21) (4,283.4)
      };
      \addplot[
        draw=red!70!black,
        fill=red!20,
        bar width=4.5pt,
        bar shift=2.7pt
      ] coordinates {
        (1,22.18) (2,47.88) (3,136.25) (4,216.7)
      };
      \legend{CUDA Core, Tensor Core}

      \nextgroupplot[
        ylabel={},
        xtick={1,2,3,4},
        xticklabels={
          {10K$\times$100K},
          {50K$\times$100K},
          {100K$\times$100K},
          {100K$\times$150K}
        },
        ymax=520,
        ytick={0,100,200,300,400,500},
        legend style={
          draw=none,
          font=\tiny,
          at={(0.03,0.97)},
          anchor=north west,
          legend columns=1
        },
        title={(b) Evaluation on multiple GPUs.},
        title style={
          at={(0.5,-0.44)},
          anchor=north,
          font=\footnotesize
        },
      ]
      \addplot[
        draw=blue!70!black,
        fill=blue!25,
        bar width=3.2pt,
        bar shift=-3.8pt
      ] coordinates {
        (1,36.18) (2,167.07) (3,373.39) (4,507)
      };
      \addplot[
        draw=red!70!black,
        fill=red!20,
        bar width=3.2pt,
        bar shift=0pt
      ] coordinates {
        (1,22.12) (2,98.4) (3,205.33) (4,340.8)
      };
      \addplot[
        draw=brown!70!black,
        fill=brown!20,
        bar width=3.2pt,
        bar shift=3.8pt
      ] coordinates {
        (1,12.9) (2,58.4) (3,110.19) (4,180.84)
      };
      \legend{1 GPU, 2 GPU, 4 GPU}

    \end{groupplot}
  \end{tikzpicture}
  \caption{Performance comparison under different hardware configurations.}
  \label{fig:combined_two_plots}
\end{figure}
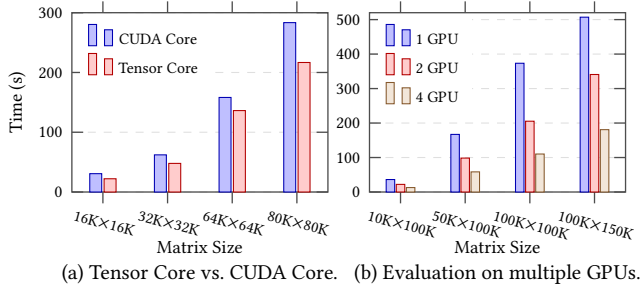

\subsection{Comparison of RMF with SVD In Stage I}
\label{sec:comparison_concise}

\begin{figure}[t]
\centering
\resizebox{0.8\columnwidth}{!}{%
\begin{tikzpicture}[]
\begin{axis}[
        scale = 0.7,
        font=\large,
        at={(-0.4,0)},
        anchor=origin,
        width=\columnwidth,
        height=0.7\columnwidth,
        xlabel={Matrix Size (K)},
        ylabel={Time (s)},
        xmin=20, xmax=100,
        ymin=0, ymax=10,
        legend style={
            at={(1.2,1.02)},
            anchor=south,
            legend columns=-1,
            /tikz/every even column/.append style={column sep=0.5cm},
            draw=none,
            font=\large
        },
        grid=major
]
\addplot[color=red, mark=square*, mark size=1.5pt] coordinates {(10, 2.0) (20, 2.2) (30, 2.2) (40, 2.1) (60, 2.1) (80, 2.1) (100, 2.1)}; 
\addplot[color=blue, mark=*, mark size=1.5pt] coordinates {(10, 3.7) (20, 3.4) (30, 3.5) (40, 4.9) (60, 6.6) (80, 7.2) (100, 8.6)}; 
 \legend{RMF (Ours), SVD}
\end{axis}
\begin{axis}[
        scale = 0.7,
        at={(6.9cm,2.6cm)}, 
        anchor=origin,
        font=\large,
       width=\columnwidth,
        height=0.7\columnwidth,
        xlabel={Matrix Size (K)},
        ylabel={NMAE},
        xmin=10, xmax=100,
        ymin=0, ymax=3.5,
        legend pos=north west,
        grid=major,
 ymode=log,
log basis y={10} 
]
\addplot[color=red, mark=square*, mark size=1.5pt] coordinates {(10, 0.9) (20, 0.05) (40, 0.008) (60, 0.005) (80, 0.004) (100, 0.00369)};
\addplot[color=blue, mark=*, mark size=1.5pt] coordinates {(10, 3.3) (20,3.32) (40, 2.7) (60, 2.71) (80, 2.7) (100, 2.66)}; 
    \end{axis}
  \end{tikzpicture}%
 }
 \caption{Time and accuracy comparison: RMF vs. SVD.}
 \label{fig:combined_convergence_scalability}
\end{figure}
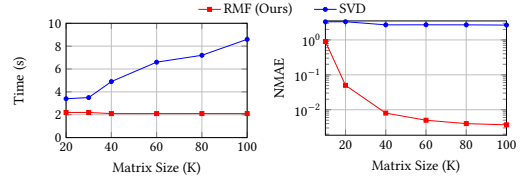

To quantify the benefit of RMF in Stage~I, we replace it with a well-optimized SVD-based implementation and compare runtime and accuracy under the same submatrix-recovery setting.
As shown in Figure~\ref{fig:combined_convergence_scalability}, our RMF-based method reduces NMAE by up to \textbf{99.6\%} and is 4.09$\times$ faster than SVD. It also converges more smoothly, quickly, and stably due to efficient parallel SGD-based RMF optimization. Moreover, as matrix size increases, our accuracy improves while SVD remains nearly unchanged, reflecting the suitability of RMF for submatrix recovery and the benefit of gradient clipping.

\subsection{Multi-GPU Scalability.}

\projectname scales across multiple GPUs. In \textbf{Stage~I}, the template submatrix $R'$ is recovered by SGD-based RMF on a single GPU. This phase is lightweight and memory-bound, so multi-GPU parallelism brings limited benefit. In \textbf{Stage~II}, \projectname exploits multi-GPU parallelism by distributing column chunks across GPUs in a chunk-based pipeline, where each GPU processes a disjoint subset and results are merged to reconstruct the matrix. Asynchronous prefetching and write-back overlap I/O with computation to maximize utilization.

We evaluate \projectname with 1, 2, and 4 GPUs (Figure~\ref{fig:combined_two_plots}) on multiple matrix sizes. Runtime decreases as GPUs increase. For the smallest problem, runtime drops from {36.18s} on 1 GPU to {22.12s} on 2 GPUs (\textbf{38.9\%}), and to {18.90s} on 4 GPUs (\textbf{47.8\%}). For a \(100\text{K}\times150\text{K}\) matrix, runtime falls from {507s} to {340.80s} (\textbf{32.8\%}), and further to {220.84s} (\textbf{56.4\%}). These results show that \projectname\ distributes work effectively across GPUs with modest overhead, and that its chunk-based design scales well to larger problems.


\subsection{Evaluation on Large-Scale Matrix}

\begin{figure}[t]
  \centering
  \resizebox{0.8\columnwidth}{!}{
    \begin{tikzpicture}[scale=1]
      \pgfmathsetmacro{\xone}{1}
      \pgfmathsetmacro{\xtwo}{2}
      \pgfmathsetmacro{\xthree}{3}
      \pgfmathsetmacro{\xfour}{4}
      
      \begin{axis}[
          name=plot1,
width=1.0\columnwidth, 
        height=0.5\columnwidth, 
          axis x line=bottom,
          axis y line=left,
          font=\small,
          ylabel={Time (s)},
          xlabel={Matrix Size},
          xtick={\xone,\xtwo,\xthree,\xfour},
          xticklabels={10K$\times$100K, 50K$\times$100K, 100K$\times$100K, 100K$\times$150K},
          xticklabel style={align=center},
          xmin=0.5, xmax=4.5,  
          ymin=0, ymax=550,
          ytick={0,100,200,300,400,500},
          grid=both,
          grid style={dashed,gray!30},
          legend style={
            at={(0.2,1)},
            font=\small,
            anchor=south west,
            legend columns=1,
            draw=none
          },
        ]
        \addplot+[
          ybar, 
          bar width=0.5cm,
          bar shift=0pt,  
          fill=blue!80!,
          thick,
          opacity=0.4,
        ] coordinates {
          (\xone, 36.18) 
          (\xtwo, 167.07) 
          (\xthree, 373.39) 
          (\xfour, 507)
        };
        \addlegendentry{Time}
      \end{axis}
      
      \begin{axis}[
          name=plot2,
          at={(plot1.south west)},
          anchor=south west,
width=1.0\columnwidth, 
        height=0.5\columnwidth, 
          axis x line=none,
          axis y line=right,
          font=\small,
          ylabel={NMAE},
          xtick=\empty,  
          xmin=0.5, xmax=4.5,  
          ymin=0, ymax=4,
          ytick={0,1,2,3,4},
          yticklabel style={/pgf/number format/fixed},
          legend style={
            at={(0.8,1)},
            anchor=south east,
            legend columns=1,
            draw=none,
            font=\small,
          },
        ]
        \addplot[
          mark=square*,
          thick,
          mark options={fill=orange},
          mark size=3pt,
        ] coordinates {
          (\xone, 2.45) 
          (\xtwo, 2.43) 
          (\xthree, 1.86) 
          (\xfour, 3.2)
        };
        \addlegendentry{Accuracy}
      \end{axis}
    \end{tikzpicture}
  }
  \caption{Evaluation on large scale matrices.}
  \label{fig:large_scale}
\end{figure}
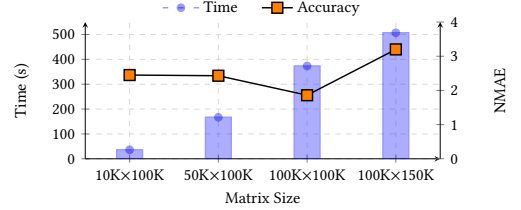

Figure~\ref{fig:large_scale} shows the execution time and NMAE of \projectname\ on large matrices, where all baselines run out of memory. Runtime increases from 36.18s on 10K$\times$100K to 507s on 100K$\times$150K, while NMAE does not degrade monotonically: it stays around 2.4 for 10K$\times$100K and 50K$\times$100K, then rises to \textbf{3.2$\times$} on 100K$\times$150K. This suggests that, under current hyper-parameters, \projectname\ performs best near the 100K$\times$100K scale.



\subsection{Parameter Sensitivity}~\label{eva:sen}

\begin{figure}[t]
\centering
\resizebox{\columnwidth}{!}{%
\begin{tikzpicture}
    \begin{axis}[
        scale=0.7,
        font=\normalsize,
        at={(-7cm,0cm)},
        anchor=origin,
        width=0.7\columnwidth,
        height=0.7\columnwidth,
        xlabel={Subcolumn Ratio},
        ylabel={Time (s)},
        xmin=0, xmax=1.0,
        ymin=0, ymax=3,
        grid=major
    ]
    \addplot[color=red, mark=square*, smooth, mark size=1.5pt] coordinates {
        (0.01, 2.49) (0.03, 2.34) (0.05, 1.83) (0.07, 1.36) 
        (0.1, 0.847) (0.3, 0.75) (0.5, 0.45) (0.7, 0.792) (1.0, 1.1)
    };
    \end{axis}
    
    \begin{axis}[
        scale=0.7,
        font=\normalsize,
        at={(-3.0cm,0cm)},
        anchor=origin,
        width=0.7\columnwidth,
        height=0.7\columnwidth,
        xlabel={Number of Chunks},
        xmin=5, xmax=30,
        ymin=0, ymax=2,
        grid=major
    ]
    \addplot[color=blue, mark=square*, smooth, mark size=1.5pt] coordinates {
        (5, 2.02) (10, 0.66) (15, 1.3) (20, 1.4) (25, 1.52) (30, 1.85)
    };
    \end{axis}

    \begin{axis}[
         scale=0.7,
        font=\normalsize,
        at={(2.0cm,0cm)},
        anchor=origin,
        width=0.7\columnwidth,
        height=0.7\columnwidth,
        xlabel={Subcolumn Ratio},
        ylabel={NMAE},
        xmin=0.05, xmax=1.0,
        ymin=0, ymax=1,
        grid=major
    ]
    \addplot[color=green, mark=square*, mark size=1.5pt] coordinates {
        (0.015, 0.0221) (0.031, 0.0258) (0.052, 0.0356) (0.073, 0.0825) (0.14, 0.1774) (0.5, 0.35) 
        (0.6, 0.3) (0.7, 0.28) (0.8, 0.26) (0.9, 0.24) (1.0, 0.22)
    };
    \end{axis}
    
    \begin{axis}[
        scale=0.7,
        font=\normalsize,
        at={(6cm,0cm)},
        anchor=origin,
        width=0.7\columnwidth,
        height=0.7\columnwidth,
        xlabel={Number of Chunks},
        xmin=0, xmax=30,
        ymin=0, ymax=0.02,
        grid=major
    ]
    \addplot[color=purple, mark=square*, mark size=1.5pt] coordinates {
        (5, 0.02) (10, 0.009) (15, 0.008) (20, 0.005) (25, 0.004) (30, 0.00369)
    };
    \end{axis}

\end{tikzpicture}%
}
\caption{Sampling ratio and chunk number sensitivity.}
\label{fig:Time}
\end{figure}
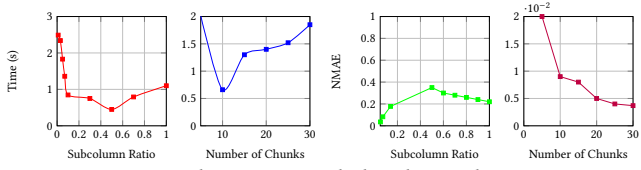

\noindent\textbf{Effect of Subcolumn Sampling Ratio.} The subcolumn sampling ratio determines the proportion of columns selected during Stage~I. As shown in Figure~\ref{fig:Time}, increasing the ratio from 0.01 to 0.1 leads to a sharp drop in runtime, i.e., from 2.49s to 0.847s, because the submatrix captures enough structure to enable faster convergence. At very low ratios (e.g., 0.01), the submatrix lacks sufficient information, resulting in more iterations and higher computational costs. From 0.1 to 0.5, runtime continues to decrease slowly, reaching a minimum of 0.45s at 0.5. This range balances representativeness and overhead. However, beyond 0.5, runtime increases again, up to 1.1s at a ratio of 1.0 due to the growing cost of processing a large submatrix. While a larger submatrix reduces Stage~II workload, the overhead of handling more data offsets these gains, ultimately reducing overall efficiency. Thus, for small datasets, a ratio of 0.5 typically yields the best performance. For large-scale datasets, however, higher ratios incur excessive memory usage; therefore, we adopt a default ratio of 0.1 to balance efficiency and memory consumption.

\noindent\textbf{Parameter in Stage I.} The number of chunks determines how the matrix is partitioned in Stage~II for full-matrix reconstruction using chunk-based pipelined ALS. While it does not affect accuracy (per TSID theory), it strongly impacts runtime by influencing memory usage and parallelism. As shown in Figure~\ref{fig:nmae}, increasing the chunk count from 5 to 10 reduces runtime from 2.02s to 0.66s. With only 5 chunks, each chunk is large, limiting parallelism and increasing memory pressure, which slows computation. At 10 chunks, the matrix is split into smaller parts, enabling better parallel processing, a reduced memory footprint per chunk, and effective I/O-computation overlap. However, further increasing the chunk count leads to an increase in runtime due to overheads such as frequent CPU-GPU data transfers and more synchronization across tasks. This diminishes the benefits of parallelism and leads to performance degradation. For the \(1\text{K} \times 10\text{K}\) dataset, 10 chunks offer the best trade-off. While the optimal count may vary with dataset size, 10 remains a robust default for balancing efficiency and scalability.

\noindent\textbf{Effect of Chunk Count.} In the worst case, submatrix sampling may degrade performance under highly nonuniform missing patterns. However, in real-world large-scale datasets, the likelihood of this occurring is very low. Our random sampling selects 10\% of columns, which reduces the chance of encountering severe missingness. To validate robustness, we performed 1,000 random sampling trials on a $1\text{K} \times 1\text{K}$ dataset. The performance variance across runs was less than 0.02\%, confirming the stability of our approach.

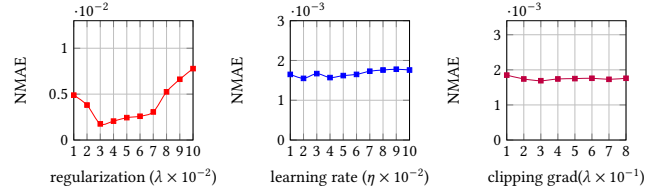
\begin{figure}[t]
\centering
\resizebox{\columnwidth}{!}{%
\begin{tikzpicture}
    \begin{axis}[
        scale=0.4,
        font=\large,
        at={(0cm,0cm)},
        anchor=origin,
        width=\columnwidth,
        height=\columnwidth,
        xlabel={regularization ($\lambda \times 10^{-2}$)}, 
        ylabel={NMAE},
        xmin=0.01, xmax=0.100,
        ymin=0, ymax=0.013,
        grid=major,
        xtick={0.01, 0.02, 0.03, 0.04, 0.05, 0.06, 0.07, 0.08, 0.09, 0.10}, 
        xticklabels={1, 2, 3, 4, 5, 6, 7, 8, 9, 10} 
    ]
    \addplot[color=red, mark=square*, smooth, mark size=1.5pt] coordinates {
        (0.01, 0.00489) (0.02, 0.00381) (0.03, 0.00176) (0.04, 0.00206) (0.05, 0.00244) (0.06, 0.00259) (0.07, 0.00306) (0.08, 0.00525) (0.09, 0.00662) (0.10, 0.00778)
    };
    \end{axis}

    \begin{axis}[
        scale=0.4,
        font=\large,
        at={(5cm,0cm)},
        anchor=origin,
        width=\columnwidth,
        height=\columnwidth,
        xlabel={learning rate ($\eta \times 10^{-2}$)}, 
        ylabel={NMAE},
        xmin=0.01, xmax=0.10,
        ymin=0, ymax=0.0030,
        grid=major,
        xtick={0.01, 0.02, 0.03, 0.04, 0.05, 0.06, 0.07, 0.08, 0.09, 0.10}, 
        xticklabels={1, 2, 3, 4, 5, 6, 7, 8, 9, 10} 
    ]
    \addplot[color=blue, mark=square*, smooth, mark size=1.5pt] coordinates {
        (0.01, 0.00165) (0.02, 0.00155) (0.03, 0.00167) (0.04, 0.00157) (0.05, 0.00162) (0.06, 0.00165) (0.07, 0.00173) (0.08, 0.00176) (0.09, 0.00178) (0.10, 0.00176)
    };
    \end{axis}

    \begin{axis}[
        scale=0.4,
        font=\large,
        at={(10cm,0cm)},
        anchor=origin,
        width=\columnwidth,
        height=\columnwidth,
        xlabel={clipping grad($\lambda \times 10^{-1}$)}, 
        ylabel={NMAE},
        xmin=0.5, xmax=1.200,
        ymin=0, ymax=0.0034,
        grid=major,
        xtick={0.5, 0.6, 0.7, 0.8, 0.9, 1.0 , 1.1 , 1.2}, 
        xticklabels={1, 2, 3, 4, 5, 6, 7, 8, 9, 10} 
    ]
    \addplot[color=purple, mark=square*, smooth, mark size=1.5pt] coordinates {
        (0.5, 0.00185) (0.6, 0.00174) (0.7, 0.00169) (0.8, 0.00174) (0.9, 0.00175) (1.0, 0.00176) (1.1, 0.00173) (1.2, 0.00176)
    };
    \end{axis}
\end{tikzpicture}%
}
\caption{NMAE sensitivity.}
\label{fig:nmae}
\end{figure}

\subsection{Limitation and Impact}
Our design assumes that the sampled template submatrix \(R^*\) fits in GPU memory. For extremely large problems, even the template may exceed device capacity, and we do not yet provide an out-of-core implementation for Stage~I. Extending the template construction to out-of-core or multi-GPU settings is an important direction for future work.

On the positive side, by reducing matrix completion to a template-resident, streamable workflow, \projectname enables billion-scale completion on a single node and improves hardware efficiency on modern accelerators. This design can benefit a range of data-driven applications, such as recommendation systems and scientific data recovery, where large sparse matrices must be processed under tight memory budgets.

\section{Related Work}
\label{sec:related}

\noindent\textbf{Matrix Completion Algorithms.}
Early theoretical work established near-exact recovery guarantees using convex nuclear-norm relaxations~\cite{candes2009exact,cai2010singular,krajewska2024matrix,shin2014calibrationless} and nonconvex low-rank factorizations of the form \(R \approx P Q^\top\)~\cite{keshavan2010matrix, koohi2019parallel,radhakrishnan2022simple}. Subsequent Bayesian formulations introduced sparsity-promoting priors over factors to improve robustness and uncertainty estimation~\cite{babacan2012sparse}. More recent work has continued to strengthen the theory of matrix completion, including nearly linear-time robust alternating minimization~\cite{gu2024robustam} and distribution-free uncertainty quantification via conformalized matrix completion~\cite{gui2023conformalized}.

\noindent\textbf{Other Matrix Completion Systems.}
To handle web-scale recommendation workloads, several systems parallelize matrix factorization across clusters or heterogeneous devices. A block-partitioned distributed SGD algorithm for matrices with billions of entries~\cite{gemulla2011large}, and distributed variants such as DALS, ASGD, and DSGD++ that exploit both data and model parallelism~\cite{teflioudi2012distributed}. More recent systems optimizie execution on heterogeneous CPU--GPU platforms~\cite{yu2021efficient}. These systems substantially accelerate matrix factorization, but they still treat matrix completion as a whole-matrix workload whose data or factors are partitioned across workers or accelerators.

\noindent\textbf{Positioning of \projectname.}
\projectname shares the same low-rank factorization objective as above work but targets a different point in the design space.
Rather than scaling out by sharding factors across a cluster, \projectname scales up by using TSID as a resource-reduction primitive: a small template submatrix is kept in GPU memory while the remaining columns are streamed and reconstructed under a fixed memory budget.
This submatrix-centric design restructures matrix completion into a bounded-memory execution workflow, enabling single-node, billion-scale completion that can be reused across applications.

\section{Conclusion}
We presented \projectname, a submatrix-centric matrix completion system that uses TSID as a resource-reduction primitive to decouple logical matrix size from the GPU working set. This enables bounded-memory completion at billion-scale and yields up to \(11{,}647\times\) speedup with lower peak memory and higher accuracy. We hope that viewing TSID as a reusable resource-reduction primitive will inspire future systems for large-scale numerical workloads.

\bibliographystyle{ACM-Reference-Format}
\bibliography{ref}




\end{document}